%% file: rrm_OMP.tex
\documentclass[english,10pt,journal]{IEEEtran}
\usepackage{color,graphicx,amsmath,amssymb,amsthm,epsfig,mathrsfs,cite,bm,enumerate}
\usepackage{array}
\usepackage{verbatim}
\usepackage{mathtools}

\DeclarePairedDelimiter\floor{\lfloor}{\rfloor}
\usepackage{color,soul}
\usepackage{multirow}
\makeatletter
\newcommand*{\rom}[1]{\expandafter\@slowromancap\romannumeral #1@}
\makeatother

\include{notation}
\theoremstyle{remark} \newtheorem{remark}{Remark}
\usepackage{amssymb}
\usepackage{float}
\usepackage{tikz}
\usetikzlibrary{trees}
\usepackage[lofdepth,lotdepth]{subfig}
\usepackage{enumerate}

\usepackage{multicol}
\usepackage{stfloats}

\title{High SNR Consistent Compressive Sensing Without Signal and Noise Statistics } 

\author{Sreejith Kallummil,  \hspace{0cm} Sheetal Kalyani
 \thanks{Sreejith Kallummil is with Samsung Adavanced Institute of Technology (SAIT-INDIA), Bangalore, India. This work was done while he was a graduate student in the Department of Electrical Engineering, IIT Madras. Email: sreejith.k.venugopal@gmail.com }
 \thanks{Sheetal Kalyani is with the Department of Electrical Engineering, Indian
Institute of Technology Madras, Chennai 600036, India. Email:  skalyani@ee.iitm.ac.in.}

}
\begin{document}
\maketitle

\begin{abstract}
Recovering the support of sparse vectors in underdetermined linear regression models, \textit{aka}, compressive sensing is important in many signal processing applications.   High SNR consistency (HSC), i.e., the ability  of a support recovery technique to correctly identify the support with increasing signal to noise ratio (SNR) is an increasingly  popular criterion to qualify the high SNR optimality of support recovery techniques. The  
HSC results available in literature for support recovery techniques applicable to  underdetermined linear regression models  like least absolute shrinkage and selection operator (LASSO), orthogonal matching pursuit (OMP) etc.  assume \textit{a priori} knowledge of noise variance  or signal sparsity.  However, both these parameters are unavailable in most practical applications. Further, it is extremely difficult to estimate noise variance  or signal sparsity in underdetermined regression models.  This limits the utility of existing HSC results. In this article, we propose two techniques, \textit{viz.}, residual ratio minimization (RRM) and residual ratio thresholding with adaptation (RRTA)  to operate  OMP algorithm without the \textit{a priroi}  knowledge of noise variance  and signal sparsity  and establish their HSC analytically and numerically.  To the best of our knowledge, these are the first  and only noise statistics oblivious algorithms to report HSC in underdetermined regression models. 
\end{abstract}

%
%
%


\section{Introduction}
Consider a linear regression model 
\begin{equation}
{\bf y}={\bf X}\boldsymbol{\beta}+{\bf w},
\end{equation} 
where ${\bf y} \in \mathbb{R}^n$ is the observation vector, ${\bf X}\in \mathbb{R}^{n \times p}$ is the $n \times p$ design matrix, $\boldsymbol{\beta}\in \mathbb{R}^{p}$ is the unknown regression vector and ${\bf w} \in \mathbb{R}^{n}$ is the noise vector. We consider a high dimensional or underdetermined scenario where the number of observations ($n$) is much less than the number of variables/predictors ($p$).  We also assume that the entries in the noise vector ${\bf w}$ are independent and identically distributed Gaussian random variables  with  mean zero and variance $\sigma^2$. Such regression models are widely studied in signal processing literature under the compressive sensing or compressed sensing paradigm \cite{eldar2012compressed}.  Subset selection in  linear regression models refers to the identification of support $\mathcal{S}=supp(\boldsymbol{\beta})=\{k:\boldsymbol{\beta}_k\neq 0\}$, where $\boldsymbol{\beta}_k$ refers to the $k^{th}$ entry of $\boldsymbol{\beta}$. Identifying supports in underdetermined or  high dimensional linear models is an ill posed problem  even in the absence of noise ${\bf w}$ unless the design matrix ${\bf X}$ satisfies  regularity conditions \cite{eldar2012compressed} like restricted isometry property (RIP), mutual incoherence property (MIC), exact recovery condition (ERC) etc. and $\boldsymbol{\beta}$ is sparse. A vector $\boldsymbol{\beta}$ is called sparse if  the cardinality of support $\mathcal{S}$ given by $k_0=card(\mathcal{S})\ll p$.  In words, only few entries of a sparse vector $\boldsymbol{\beta}$ will be  non-zero.   Identification of sparse supports in underdetermined linear regression models have  many applications including and not limited to detection in multiple input multiple output (MIMO)\cite{choi2017detection} and generalised MIMO systems\cite{yu2012compressed,kallummil2016combining}, multi user detection\cite{shim2012multiuser}, subspace clustering\cite{you2016scalable} etc.  This article discusses this important problem of recovering sparse supports in high dimensional linear regression models. After presenting the notations used in this article, we provide a brief summary of sparse support recovery techniques discussed in literature and the exact problem discussed in this article. 
\subsection{Notations used}
$col({\bf A})$ the column space of matrix ${\bf A}$. ${\bf A}^T$ is the transpose and ${\bf A}^{\dagger}=({\bf A}^T{\bf A})^{-1}{\bf A}^T$ is the  Moore-Penrose pseudo inverse of ${\bf A}$. ${\bf P}_{\bf A}={\bf A}{\bf A}^{\dagger}$ is the projection matrix onto $col({\bf A})$. ${\bf A}_{\mathcal{J}}$ denotes the sub-matrix of ${\bf A}$ formed using  the columns indexed by $\mathcal{J}$.  When ${\bf A}$ is clear from the context, we use the shorthand ${\bf P}_{\mathcal{J}}$ for ${\bf P}_{{\bf A}_{\mathcal{J}}}$. Both ${\bf a}_{\mathcal{J}}$ and ${\bf a}(\mathcal{J})$ denote the  entries of vector ${\bf a}$ indexed by $\mathcal{J}$.  $\mathcal{N} ({\bf u},{\bf C})$ is a Gaussian random  vector (R.V) with mean ${\bf u}$ and covariance ${\bf C}$. $\mathbb{B}(a,b)$ represents a Beta R.V with parameters $a$ and $b$ and $B(a,b)$ represents the Beta function. $F_{a,b}(x)=\dfrac{1}{B(a,b)}\int_{t=0}^x t^a(1-t)^b$ is the CDF of a $\mathbb{B}(a,b)$ R.V.   ${\bf a}\sim{\bf b}$ implies that  R.Vs ${\bf a}$ and ${\bf b}$ are identically distributed.  
$\|{\bf a}\|_m=(\sum\limits_{j}|{\bf a}_j|^m)^{\frac{1}{m}}$  for $1\leq m\leq \infty$ is the $l_m$ norm and $\|{\bf a}\|_0=card(supp({\bf a}))$ is the $l_0$ quasi norm of ${\bf a}$. 
For any two index sets $\mathcal{J}_1$ and $\mathcal{J}_2$, the set difference  $\mathcal{J}_1/\mathcal{J}_2=\{j \in \mathcal{J}_1: j\notin  \mathcal{J}_2\}$.  $X\overset{p}{\rightarrow } Y$ denotes the  convergence of random variable $X$ to $Y$ in probability. $\mathbb{P}()$ and $\mathbb{E}()$ represent probability  and expectation. Signal to noise ratio (SNR) for the regression model (1) is given by $SNR=\dfrac{\mathbb{E}(\|{\bf X}\boldsymbol{\beta}\|_2^2)}{\mathbb{E}(\|{\bf w}\|_2^2)}=\dfrac{\|{\bf X}\boldsymbol{\beta}\|_2^2}{n\sigma^2}$. 

\subsection{High SNR consistency in linear regression} \label{intro}
The quality of a  support selection technique  delivering a support estimate $\hat{\mathcal{S}}$ is typically quantified in terms of  the   probability of  support recovery error   $PE=\mathbb{P}(\hat{\mathcal{S}}\neq \mathcal{S})$ or the probability of correct support recovery  $PCS=1-PE$. The high SNR behaviour (i.e. behaviour as $\sigma^2 \rightarrow 0$ or $SNR\rightarrow \infty$) of support recovery techniques in general and  the concept of high SNR consistency (HSC) defined below  in particular has attracted considerable attention in statistical signal processing community recently\cite{ding2011inconsistency,schmidt2012consistency,stoica2012proper,stoica2013model,SNLShighSNR,tsp,spl,elsevier}. 

{\bf Definition 1:-}  A support recovery technique is defined to be high SNR consistent (HSC) iff $\underset{\sigma^2 \rightarrow 0}{\lim}PE=0$ or equivalently  $\underset{SNR \rightarrow \infty}{\lim}PE=0$.

 In applications where the problem size $(n,p)$ is small and constrained, the support recovery performance can be improved only by increasing the SNR. This makes HSC and high SNR behaviour in general very important in certain practical applications.

Most of the existing literature on HSC deal with overdetermined ($n>p$) or low dimensional linear regression models. In this context, high SNR consistent model order selection techniques like exponentially embedded family (EEF)\cite{ding2011inconsistency}\cite{eefenumeration}, normalised minimum description length (NMDL)\cite{schmidt2012consistency}, forms of Bayesian information criteria (BIC)\cite{stoica2013model}, penalised adaptive likelihood (PAL)\cite{stoica2013model}, sequentially normalised least squares (SNLS)\cite{SNLShighSNR} etc.  when combined with a $t$-statistics based variable  ordering scheme were shown to  be HSC \cite{tsp}. Likewise, the necessary and sufficient conditions (NSC) for the HSC of threshold based support recovery schemes were derived in \cite{spl}.   However, both these HSC support recovery procedures are applicable only to overdetermined ($n>p$) regression models and are not applicable to the underdetermined ($n<p$) regression  problem discussed in this article.  Necessary and sufficient conditions  for the high SNR consistency of 
compressive sensing algorithms like OMP\cite{tropp2004greed,cai2011orthogonal,tropp2007signal} and variants of LASSO \cite{tropp2006just} are derived in \cite{elsevier}.  However, for HSC  and good finite SNR estimation performance, both OMP and LASSO require either the \textit{a priori} knowledge of noise variance $\sigma^2$ or sparsity level $k_0$. Both these quantities are unknown \textit{a priori} in most practical applications.  However, unlike the case of overdetermined regression models where unbiased estimates of $\sigma^2$  with explicit finite  sample guarantees are available, no estimate of $\sigma^2$ with such finite sample guarantees are available in underdetermined regression models to the best of our knowledge.  Similarly, we are also not aware of any technique to efficiently estimate the sparsity level $k_0$.    Hence, the application of HSC results in \cite{elsevier} to practical underdetermined support recovery problems are  limited. 
\subsection{Contribution of this article}
{ Residual ratio thresholding (RRT)\cite{icml,robust,mos}    is a concept recently introduced  to perform sparse variable selection in linear regression models without the \textit{a priori} knowledge of nuisance parameters like noise variance, sparsity level etc. This concept was initially developed to operate support recovery algorithms like OMP, orthogonal least squares (OLS) etc, in underdetermined linear regression models with explicit finite SNR and finite sample guarantees \cite{icml}. Later, this concept was  extended to outlier detection problems in  robust regression \cite{robust} and model order selection in overdetermined linear regresssion \cite{mos}.   A significant drawback of RRT in the context of support recovery in underdetermined regression models (as we establish in this article) is that it is inconsistent at high SNR. In other words, inspite of having a decent finite SNR performance, RRT is suboptimal in the high SNR regime. } In this article, we propose two variants of RRT, \textit{viz.}, residual ratio minimization (RRM) and residual ratio thresholding with adaptation (RRTA) to  operate algorithms like OMP, OLS etc. without the \textit{a priori} knowledge of $k_0$ or $\sigma^2$.  Unlike RRT, these two schemes are shown to be high SNR consistent both analytically and numerically.  In addition to HSC which is an asymptotic result, we also derive finite sample and finite SNR support recovery guarantees for RRM based on RIP. These support recovery results indicate that the SNR required for successfull support recovery using RRM increases with the dynamic range of $\boldsymbol{\beta}$ given by $DR(\boldsymbol{\beta})=\dfrac{\boldsymbol{\beta}_{max}=\underset{j \in \mathcal{S}}{\max}|\boldsymbol{\beta}_j|}{\boldsymbol{\beta}_{min}=\underset{j \in \mathcal{S}}{\min}|\boldsymbol{\beta}_j|}$, whereas, numerical simulations indicate that the SNR required by RRTA (like RRT and OMP with \textit{a priori} knowledge of $\sigma^2$ or $k_0$) depends only on the minimum non zero value $\boldsymbol{\beta}_{min}$. Consequently, the finite SNR utility of RRM is limited to wireless communication applications like \cite{kallummil2016combining} where $DR(\boldsymbol{\beta})$ is close to one. In contrast to RRM, RRTA is useful in both finite and  high SNR applications irrespective of the dynamic range of $\boldsymbol{\beta}$. 
\subsection{Organization of this article}
Section \rom{2} presents the existing results  on OMP. Section \rom{3} introduces RRT and develope  RRM and RRTA techniques along with their analytical guarantees.   Section \rom{4} presents numerical simulations.      
\section{High SNR consistency of  OMP with \textit{a priori} knowledge of $\sigma^2$ or $k_0$}
\begin{table}
\begin{tabular}{|l|}
\hline
{\bf Input:} Observation ${\bf y}$, design matrix ${\bf X}$ and stopping condition.\\
  {\bf Step 1:-} Initialize the residual ${\bf r}^{(0)}={\bf y}$. \\ \ \ \ \ \ \ \ \ $\hat{\boldsymbol{\beta}}={\bf 0}_p$,    Support estimate ${\mathcal{S}_0}=\emptyset$, Iteration counter $k=1$; \\
 {\bf Step 2:-} Update support estimate: ${\mathcal{S}_k}={\mathcal{S}_{k-1}}\cup t^k$, \\ \ \ \ \ \ \ \ \ \ where  $t^k=\underset{t \in [p]}{\arg\max}|{\bf X}_t^T{\bf r}^{k-1}|.$ \\
 {\bf Step 4:-} Estimate $\boldsymbol{\beta}$ using current support:\\ $ \ \ \ \ \ \ \ \ \ \hat{\boldsymbol{\beta}}(\mathcal{S}_k)={\bf X}_{\mathcal{S}_k}^{\dagger}{\bf y}$. \\
  {\bf Step 5:-} Update residual: ${\bf r}^{k}={\bf y}-{\bf X}\hat{\boldsymbol{\beta}}=({\bf I}_n-{\bf P}_{k}){\bf y}$. \\
  
\ \ \ \ \ \ \ \ \ \ \ \ \ \ \   ${\bf P}_k={\bf X}_{\mathcal{S}_k}{\bf X}_{\mathcal{S}_k}^{\dagger}$. \\
  {\bf Step 6:-} Increment $k$. $k \leftarrow k+1$. \\
  {\bf Step 7:-} Repeat Steps 2-6, until the stopping condition  is  satisfied. \\
  {\bf Output:-} Support estimate $\hat{\mathcal{S}}=\mathcal{S}_k$ and signal estimate $\hat{\boldsymbol{\beta}}$. \\
 \hline
\end{tabular}
\caption{ OMP algorithm.}
\label{tab:omp}
\end{table}
OMP \cite{tropp2004greed} in TABLE \ref{tab:omp} is a widely used greedy and iterative sparse support recovery algorithm.   OMP  algorithm starts with a null set as support estimate and observation ${\bf y}$ as the initial residual. At each iteration, OMP identifies the column  that is the most correlated with the current residual ($t^k=\underset{j}{\arg\max}|{\bf X}_j^T{
\bf r}^{k-1}|$) and expand the support estimate by including this selected column index $(\mathcal{S}_k=\mathcal{S}_{k-1} \cup t^k)$. Later, the residual is updated by projecting the observation vector ${\bf y}$ orthogonal to the column space produced by the current support estimate (i.e., $col({\bf X}_{\mathcal{S}_k})$).  Since ${\bf r}^k$ is orthogonal to the column space of ${\bf X}_{\mathcal{S}_k}$, ${\bf X}_t^T{\bf r}^k=0$ for all $t\in \mathcal{S}_k$. Consequently, an index selected in an initial stage will not be selected again later. Consequently, the support estimate sequence monotonically increases with iteration $k$, i.e., $\mathcal{S}_k\subset \mathcal{S}_{k+1}$ and $card(\mathcal{S}_k)=k$.  
\begin{remark}
OLS iterations are also similar to that of OMP except that  OLS select the column that results in the maximum decrease  in residual energy $\|{\bf r}^k\|_2^2$, i.e., $t^k=\underset{j}{\arg\min}\|({\bf I}_n-{\bf P}_{\mathcal{S}_{k-1}\cup j}){\bf y}\|_2^2$. OLS support estimate sequence also satisifes $\mathcal{S}_k\subset \mathcal{S}_{k+1}$ and $card(\mathcal{S}_k)=k$. The techniques developed in this article will be discussed using OMP algorithm. However, please note that these techniques are equally applicable to OLS also. 
\end{remark}
The  iterations in OMP are continued  until a user defined stopping condition is met.  The performance of OMP  depends crucially on this stopping condition.  When the sparsity level $k_0$ is known \textit{a priori},  many articles suggest stopping OMP exactly after $k_0$ iterations. When $k_0$ is unknown \textit{ a priori}, one can stop OMP when the residual power $\|{\bf r}^k\|_2$ is sufficiently small. Two such residual based stopping conditions are popular in literature\cite{cai2011orthogonal}. One rule proposes to stop OMP iterations once the residual power drops below $\|{\bf r}^k\|_2\leq \|{\bf w}\|_2$, whereas, another rule proposes to stop OMP when the residual correlation drops below $\|{\bf X}^T{\bf r}^k\|_{\infty}\leq \|{\bf X}^T{\bf w}\|_{\infty}$. When  ${\bf w} \sim \mathcal{N}({\bf 0}_n,\sigma^2{\bf I}_n)$ and the columns ${\bf X}_j$ have unit $l_2$ norm, it was shown in \cite{cai2011orthogonal} that 
\begin{equation}
\begin{array}{ll}
\mathbb{P}\left(\|{\bf w}\|_2\geq \sigma\sqrt{n+2\sqrt{n\log(n)}}\right)\leq 1/n
\ \text{and}\\ 
\mathbb{P}\left(\|{\bf X}^T{\bf w}\|_{\infty}\geq \sigma\sqrt{2\log(p)}\right)\leq 1/p.
\end{array}
\end{equation}
Consequently, one can stop OMP iterations in Gaussian noise once $\|{\bf r}^k\|_2\leq \sigma\sqrt{n+2\sqrt{n\log(n)}}$ or  $\|{\bf X}^T{\bf w}\|_{\infty}\leq \sigma\sqrt{2\log(p)}$.

A number of deterministic recovery guarantees are proposed for OMP. Among these guarantees, the conditions based on restricted isometry constants (RIC) are the most popular for OMP.  RIC  of order $j$ denoted by $\delta_j$ is defined as the smallest value of $\delta$ such that 
\begin{equation} 
(1-\delta)\|{\bf b}\|_2^2\leq \|{\bf X}{\bf b}\|_2^2\leq (1+\delta)\|{\bf b}\|_2^2
\end{equation}
 hold true for all ${\bf b} \in \mathbb{R}^p$ with $\|{\bf b}\|_0=card(supp({\bf b}))\leq j$. A smaller value of $\delta_j$ implies that ${\bf X}$ act as a near orthogonal matrix for all $j$ sparse vectors ${\bf b}$. Such a situation is ideal for the recovery of a $j$-sparse vector ${\bf b}$ using any sparse recovery technique. The latest RIC based finite SNR support recovery guarantee and HSC results for OMP are given in Lemma \ref{lemma:latest_omp}. 
\begin{lemma}\label{lemma:latest_omp}
 Suppose that the matrix ${\bf X}$ satisfies $\delta_{k_0+1}<{1}/{\sqrt{k_0+1}}$. Then, \\
1).  OMP with $k_0$ iterations or stopping condition $\|{\bf r}^k\|_2\leq \|{\bf w}\|_2$ can recover any $k_0$ sparse vector $\boldsymbol{\beta}$ once $\|{\bf w}\|_2\leq  \epsilon_{omp}=\boldsymbol{\beta}_{min}\sqrt{1-\delta_{k_0+1}}\left[\dfrac{1-\sqrt{k_0+1}\delta_{k_0+1}}{1+\sqrt{1-\delta_{k_0+1}^2}-\sqrt{k_0+1}\delta_{k_0+1}}\right]$ \cite{latest_omp}. \\
2). Define $\epsilon_{\sigma}=\sigma\sqrt{n+2\sqrt{n\log(n)}}$. Then, OMP with $k_0$ iterations or stopping condition $\|{\bf r}^k\|_2\leq \epsilon_{\sigma}$ can recover any $k_0$ sparse vector $\boldsymbol{\beta}$ with a probability greater than $1-1/n$ once $\epsilon_{\sigma}\leq  \epsilon_{omp}$. \\
3). OMP running precisely $k_0$ iterations is high SNR consistent, i.e., $\underset{\sigma^2\rightarrow 0}{\lim}\mathbb{P}(\mathcal{S}_{k_0}=\mathcal{S})=1$ \cite{elsevier}. \\
4). OMP with stopping rule $\|{\bf r}^k\|_2 \leq \sigma g(\sigma) $ is HSC iff $\underset{\sigma^2\rightarrow 0}{\lim}g(\sigma)=\infty$ and  $\underset{\sigma^2\rightarrow 0}{\lim}\sigma g(\sigma)=0$ \cite{elsevier}. \\
5). OMP with stopping rule $\|{\bf X}^T{\bf r}^k\|_{\infty} \leq \sigma g(\sigma) $ is HSC iff $\underset{\sigma^2\rightarrow 0}{\lim}g(\sigma)=\infty$ and  $\underset{\sigma^2\rightarrow 0}{\lim}\sigma g(\sigma)=0$ \cite{elsevier}.

\end{lemma}

  Lemma \ref{lemma:latest_omp} implies that OMP with the \textit{a priori} knowledge of $k_0$ or $\sigma^2$ can recover support $\mathcal{S}$ once the matrix satisfies the regularity condition  $\delta_{k_0+1}<{1}/{\sqrt{k_0+1}}$ and the SNR is sufficiently high.  Lemma \ref{lemma:latest_omp} also implies that OMP with \textit{ a priori} knowledge of $k_0$ is always HSC. Further, stopping conditions $\|{\bf r}^{(k)}\|_2<\sigma\sqrt{n+2\sqrt{n\log(n)}}$ \cite{cai2011orthogonal,omp_rip_noise} or $\|{\bf X}^T{\bf r}^{(k)}\|_{\infty}<\sigma \sqrt{2\log(p)}$\cite{cai2011orthogonal} which fail to satisfy 4) and 5) of Lemma \ref{lemma:latest_omp}  are inconsistent at high SNR.  


\section{Residual ratio techniques}
As one can see from Lemma \ref{lemma:latest_omp},  good finite SNR support recovery guarantees and HSC using OMP require either the \textit{ a priori} knowledge of  $k_0$ or $\sigma^2$. However,  as mentioned earlier, both $k_0$ and $\sigma^2$ are not available in most practical applications.  Recently, we demonstrated in \cite{icml} that one can achieve high quality support recovery using OMP  without the \textit{a priori} knowledge of $\sigma^2$ or $k_0$ by using the properties of residual ratio statistic defined by $RR(k)=\dfrac{\|{\bf r}^k\|_2}{\|{\bf r}^{k-1}\|_2}$, where ${\bf r}^k=({\bf I}_n-{\bf P}_{k}){\bf y}$ is the residual corresponding to OMP support at the $k^{th}$ iteration, i.e., $\mathcal{S}_k$.  The technique developed in \cite{icml} was based on the behaviour of $RR(k)$ for $k=1,2,\dotsc,k_{max}$, where $k_{max}\geq k_0$ is a fixed quantity independent of data. $k_{max}$ is a measure of the maximum sparsity level  expected in a support recovery  experiment.  Since  the maximum sparsity level upto which support recovery can be guaranteed for any sparse recovery algorithm (not just OMP) is $\floor{\dfrac{n+1}{2}}$, \cite{icml} suggests fixing $k_{max}=\floor{\dfrac{n+1}{2}}$. Note that this is a fixed value that is independent of the data and the algorithm (OMP or OLS) under consideration.  The residual ratio statistic   has many interesting properties as derived in \cite{icml}. Since the support sequence is monotonic i.e., $\mathcal{S}_k\subset \mathcal{S}_{k+1}$,  the residual ${\bf r}^k$ is obtained by projecting ${\bf y}$ onto a subspace of decreasing dimension. Hence, $\|{\bf r}^{k+1}\|_2 \leq \|{\bf r}^k\|_2$ which inturn implies  that $0\leq RR(k)\leq 1$. Please note that while residual norms are monotonically decreasing, residual ratios $RR(k)$ are not monotonic in $k$.  A number of properties regarding the residual ratio statistic are based on the concept of minimal superset.     

{\bf Definition 2:-} 
The minimal superset in the OMP support sequence $\{\mathcal{S}_{k}\}_{k=1}^{k_{max}}$ is  given by $\mathcal{S}_{k_{min}}$, 
where $k_{min}=\min\left(\{k:\mathcal{S}\subseteq \mathcal{S}_k\}\right)$. When the set $\{k:\mathcal{S}\subseteq \mathcal{S}_k\}=\emptyset$, we set $k_{min}=\infty$ and $\mathcal{S}_{k_{min}}=\phi$. 

In words, minimal superset is the smallest superset of  support $\mathcal{S}$ present in a particular realization of the support estimate sequence $\{\mathcal{S}_{k}\}_{k=1}^{k_{max}}$. Note that both $k_{min}$ and $\mathcal{S}_{k_{min}}$ are unobservable random variables. Since $card(\mathcal{S}_k)=k$ and $card(\mathcal{S})=k_0$, $\mathcal{S}_{k}$ for $k<k_0$ cannot satisfy $\mathcal{S}\subseteq \mathcal{S}_k$ and hence  $k_{min}\geq k_0$. Further,  the  monotonicity of $\mathcal{S}_k$ implies that   $\mathcal{S} \subset \mathcal{S}_k$ for all $k\geq k_{min}$. \\
{\bf Case 1:-}  When $k_{min}=k_0$, then $\mathcal{S}_{k_0}=\mathcal{S}$ and $\mathcal{S}_{k}\supset \mathcal{S}$ for $k\geq k_0$, i.e., $\mathcal{S}$ is present in the solution path. Further, when $k_{min}=k_0$, it is true that $\mathcal{S}_k\subseteq \mathcal{S} $ for $k\leq k_0$. \\
{\bf Case 2:-} When $k_0<k_{min}\leq k_{max}$, then $\mathcal{S}_{k}\neq  \mathcal{S}$ for all $k$ and  $\mathcal{S}_{k}\supset \mathcal{S}$ for $k\geq k_{min}$, i.e., $\mathcal{S}$ is not present in the solution path. However, a superset of $\mathcal{S}$ is present. \\
{\bf Case 3:-} When $k_{min}=\infty$, then $\mathcal{S}_{k}\not \supseteq   \mathcal{S}$ for all $k$, i.e., neither $\mathcal{S}$ nor a superset of $\mathcal{S}$  is present in $\{\mathcal{S}_{k}\}_{k=1}^{k_{max}}$. \\
To summarize,   exact support recovery using any  OMP/OLS based scheme   is possible only if $k_{min}=k_0$. Whenever $k_{min}>k_0$, it is possible to estimate true support $\mathcal{S}$ without having any  false negatives. However, one then has to suffer from false positives. When $k_{min}=\infty$, any support in $\{\mathcal{S}_{k}\}_{k=1}^{k_{max}}$ has to suffer from false negatives and all supports $\mathcal{S}_{k}$ for $k>k_0-1$ has to suffer from false positives also.  Note that the matrix and SNR conditions required for exact support recovery  in statements 1) and 2) of Lemma \ref{lemma:latest_omp}  implies that $k_{min}=k_0$ and $\mathcal{S}_{k_0}=\mathcal{S}$ at high SNR. The main distributional properties of residual ratio statistic are stated in the following lemma \cite{icml}. 
\begin{lemma}\label{lemma:RR_properties}
1). Define $\Gamma_{RRT}^{\alpha}(k)=\sqrt{F^{-1}_{\frac{n-k}{2},\frac{1}{2}}\left(\dfrac{\alpha}{k_{max}(p-k+1)}\right)}$, where  $F^{-1}_{a,b}(.)$ is the inverse function of the CDF $F_{a,b}(.)$ of a Beta R.V $\mathbb{B}(a,b)$. $0<\alpha<1$ is a fixed quantity independent of the data. Then under no assumption on the matrix ${\bf X}$ and  for all $\sigma^2>0$, $RR(k)$ satisfies the following. 
\begin{equation}
\mathbb{P}(RR(k)>\Gamma_{RRT}^{\alpha}(k),\forall k>k_{min})\geq 1-\alpha.
\end{equation} 
2). Suppose that the design matrix ${\bf X}$ satisfies a regularity condition  which ensures that $k_{min}=k_0$ once $\|{\bf w}\|_2\leq \epsilon$ for some $\epsilon>0$ (for example, $\delta_{k_0+1}<1/\sqrt{k_0+1}$ and $\|{\bf w}\|_2\leq \epsilon_{omp}$ in Lemma \ref{lemma:latest_omp}).  Then, 
\begin{equation}
\mathbb{P}(k_{min}=k_0)=\mathbb{P}(\mathcal{S}_{k_{min}}=\mathcal{S})\rightarrow 1\ \text{as}\ \sigma^2\rightarrow 0 \ \text{and}
 \end{equation} 
 \begin{equation}
RR(k_0)\overset{P}{\rightarrow} 0\ \text{as}\ \sigma^2\rightarrow 0.
 \end{equation} 
\end{lemma}
Lemma \ref{lemma:RR_properties} implies that under appropriate matrix conditions and sufficiently high SNR, $RR(k)$ for $k>k_0$ will be higher than the positive quantity $\Gamma_{RRT}^{\alpha}(k)$ with a high probability, whereas, $RR(k_0)$ will be smaller than $\Gamma_{RRT}^{\alpha}(k_0)$. Consequently, the last index for which $RR(k)<\Gamma_{RRT}^{\alpha}(k)$ will be equal to the sparsity level $k_0$. This motivates the RRT support estimate given by
\begin{equation}
\mathcal{S}_{RRT}=\mathcal{S}_{k_{RRT}} \ \text{where} \ k_{RRT}=\max\{k: RR(k)<\Gamma_{RRT}^{\alpha}(k)\}.
\end{equation}  
The performance guarantees for RRT are stated in Lemma \ref{lemma:icml}. 
\begin{lemma}\label{lemma:icml}
Let $k_{max}\geq k_0$ and suppose that the matrix ${\bf X}$ satisfies $\delta_{k_0+1}<\frac{1}{\sqrt{k_0+1}}$. Then \cite{icml},\\
1). RRT can recover the true support $\mathcal{S}$ with probability greater than $1-1/n-\alpha$ provided that $\epsilon_{\sigma}<\min(\epsilon_{omp},\epsilon_{rrt})$, where $\epsilon_{omp}$ is given in Lemma \ref{lemma:latest_omp} and
\begin{equation}
\epsilon_{rrt}=\dfrac{\Gamma_{RRT}^{\alpha}(k_0)\sqrt{1-\delta_{k_{0}}}\boldsymbol{\beta}_{min}}
{1+\Gamma_{RRT}^{\alpha}(k_0)}.
\end{equation}
2). $\underset{\sigma^2 \rightarrow 0}{\lim}\mathbb{P}(\mathcal{S}_{k_{RRT}}\neq \mathcal{S})\leq \alpha.$ \\
3). $\mathbb{P}(\mathcal{S}_{k_{RRT}}\supseteq \mathcal{S})\leq \alpha$ in the moderate to high  SNR regime (empirical result).  
\end{lemma}
Lemma \ref{lemma:icml} implies that RRT can identify the true support $\mathcal{S}$ under the same set of matrix conditions required by OMP with \textit{a priori} knowledge of $k_0$ or $\sigma^2$, albeit at a slightly higher SNR. Further,  the probability of support recovery error is upper bounded by $\alpha$ at high SNR. Even in the low to moderately high SNR regime, empirical results indicate that the probability of false discovery (i.e., $card(\mathcal{S}_{k_{RRT}}/\mathcal{S})>0$) is upper bounded by $\alpha$. Hence, in RRT, the hyper parameter $\alpha$ has an operational interpretation of being the high SNR support recovery error and finite SNR false discovery error.  However, no lower bound on the  probability of support recovery error at high SNR is reported in \cite{icml}. In the following lemma, we establish a novel  high SNR  lower bound on the  probability of support recovery error for RRT. 
  \begin{lemma}\label{lemma:inconsistency}
Suppose that the design matrix ${\bf X}$ satisfies the RIC condition $\delta_{k_0+1}<1/\sqrt{k_0+1}$, MIC or the ERC. Then 
\begin{equation}
\underset{\sigma^2\rightarrow 0}{\lim}\mathbb{P}(\mathcal{S}_{RRT}\supset \mathcal{S})\geq \dfrac{\alpha}{k_{max}(p-k_0)}.
\end{equation}
  \end{lemma}
  \begin{proof} Please see Appendix A. 
  \end{proof}
In words, Lemma \ref{lemma:inconsistency} implies that RRT is inconsistent at high SNR.  { In particular, Lemma \ref{lemma:inconsistency}  states that RRT suffers from  false discoveries at high SNR. Indeed, one can reduce the lower bound  on support recovery error by reducing the value of $\alpha$. Since $\Gamma_{RRT}^{\alpha}(k_0)$ is an increasing function of $\alpha$\cite{icml}, reducing the value of $\alpha$ will result in decrease in $\epsilon_{rrt}$. Hence, a decrease in $\alpha$ to reduce the high SNR will result in an increase in the SNR required for accurate support recovery according to Lemma \ref{lemma:icml}. In other words, it is impossible  to improve the high SNR performance in RRT without compromising on the finite SNR performance. This is because of the fact that $\alpha$ is a user defined parameter that has to be  set independent of the data. A good solution would be to use a value of $\alpha$ like $\alpha=0.1$ ( recommended in \cite{icml} based on finite SNR estimation performance) for low to moderate SNR  and a low value of $\alpha$ like $\alpha=0.01$  or $\alpha=0.001$ in the high SNR regime. Since it is impossible to estimate SNR or $\sigma^2$ in underdetermined linear models, the statistician is unaware of operating SNR and cannot make such adaptations on $\alpha$.  Hence, achieving very low values of PE at high SNR  or HSC using  RRT is extremely difficult. }  This motivates the novel  RRM and RRTA algorithms discussed next which can achieve HSC using the residual ratio statistic itself. 
\subsection{Residual ratio minimization}
The analysis of $RR(k)$ in \cite{icml} (see Lemma \ref{lemma:RR_properties}) discussed only the behaviour of $RR(k)$ for $k\geq k_{min}$. However, no analysis of $RR(k)$ for $k<k_{min}$ is mentioned in \cite{icml}. In the following lemma, we charecterize the behaviour of $RR(k)$ for $k<k_{min}$. 
\begin{lemma} \label{lemma:kless}
Suppose that the matrix ${\bf X}$ satisfies the RIC condition $\delta_{k_0+1}<1/\sqrt{k_0+1}$. Then, for $k<k_0$,
\begin{equation}
 \underset{\sigma^2\rightarrow 0}{\lim}\mathbb{P}\left(RR(k)>\dfrac{\sqrt{1-\delta_{k_0}}\boldsymbol{\beta}_{min}}{\sqrt{1+\delta_{k_0}}(\boldsymbol{\beta}_{max}+\boldsymbol{\beta}_{min})}\right)=1.
 \end{equation}
 
\end{lemma}
\begin{proof}
Please see Appendix B.
\end{proof}
Combining Lemma \ref{lemma:RR_properties} and Lemma \ref{lemma:kless}, one can see that with increasing SNR, $RR(k)$ for $k<k_0$  is bounded away from zero, whereas,  $RR(k)$ for $k>k_0$ behave like a R.V that is bounded from below by a constant with a very high probability. In contrast to the behaviour of $RR(k)$ for $k\neq k_0$,  $RR(k_0)$ converges to zero with increasing SNR. Consequently, under appropriate regularity conditions on the design matrix ${\bf X}$, $\underset{k}{\arg\min}RR(k)$ will converge to $k_0$  with increasing  SNR. Also from Lemmas \ref{lemma:latest_omp} and \ref{lemma:RR_properties},  we know that $k_{min}=k_0$ and $\mathcal{S}_{k_0}=\mathcal{S}$ with a very high probability  at high SNR.  Consequently,  the  support estimate given by 
\begin{equation}
\mathcal{S}_{RRM}=\mathcal{S}_{k_{RRM}}, \ \text{where} \ k_{RRM}= \underset{k=1,\dotsc,k_{max}}{\arg\min}RR(k)
\end{equation}
will be equal to the true support $\mathcal{S}$ with a probability increasing with increasing SNR.  $\mathcal{S}_{RRM}$ is the residual ratio minimization  based support estimate proposed in this article. 
The following theorem states that RRM is a high SNR consistent estimator of support $\mathcal{S}$.
\begin{thm}\label{thm:rmm_hsc}
Suppose that the matrix ${\bf X}$ satisfies RIC condition $\delta_{k_0+1}<\dfrac{1}{\sqrt{k_0+1}}$. Then RRM is high SNR consistent, i.e., $\underset{\sigma^2 \rightarrow 0}{\lim}\mathbb{P}(\mathcal{S}_{RRM}=\mathcal{S})=1$. 
\end{thm}
\begin{proof}
Please see Appendix D.
\end{proof}

While HSC is an important qualifier for any support recovery technique, it's finite SNR performance is also very important. The following theorem quantifies the finite SNR performance of RRM. 
\begin{thm}\label{thm:rrm_finite}
Suppose that the design matrix ${\bf X}$ satisfies $\delta_{k_{0}+1}<\dfrac{1}{\sqrt{k_{0}+1}}$ and $0<\alpha<1$ is a constant. Then RRM can recover the true support $\mathcal{S}$ with a probability greater than $1-1/n-\alpha$ once $\epsilon_{\sigma}=\sigma\sqrt{n+2\sqrt{n\log(n)}}<\min\left(\epsilon_{omp},\tilde{\epsilon_{rrt}},\epsilon_{rrm}\right)$, where $\epsilon_{omp}$ is given in Lemma \ref{lemma:latest_omp}, 
\begin{equation}\label{epsb}
{\epsilon_{rrm}}= \dfrac{\sqrt{1-\delta_{k_{0}}}\boldsymbol{\beta}_{min}}{1+\frac{\sqrt{1+\delta_{k_{0}}}}{\sqrt{1-\delta_{k_{0}}}} \left(2+  \frac{\boldsymbol{\beta}_{max}}{\boldsymbol{\beta}_{min}} \right)} \ \text{and}
\end{equation}  
 \begin{equation}\label{epsb}
\tilde{\epsilon}_{rrt}=\dfrac{\underset{1\leq k\leq k_{max}}{\min}\Gamma_{RRT}^{\alpha}(k)\sqrt{1-\delta_{{k_{0}}}}\boldsymbol{\beta}_{min}}
{1+\underset{1\leq k\leq k_{max}}{\min}\Gamma_{RRT}^{\alpha}(k)}. 
\end{equation}
\end{thm}
\begin{proof} Please see Appendix C.
\end{proof}
Please note that the presence of $\alpha$ in Theorem \ref{thm:rrm_finite}  is an artefact of our analytical framework. Very importantly, unlike RRT, there are no user specified hyperparameters in RRM.
The following remark compares the finite SNR support recovery guarantee for RRM  in Theorem \ref{thm:rrm_finite} with that of RRT.
\begin{remark} For the same $(1-\alpha-1/n)$ bound on the probability of error, RRM requires higher SNR level than RRT. This is true since $\tilde{\epsilon}_{rrt}=\dfrac{\underset{1\leq k\leq k_{max}}{\min}\Gamma_{RRT}^{\alpha}(k)\sqrt{1-\delta_{{k_{0}}}}\boldsymbol{\beta}_{min}}
{1+\underset{1\leq k\leq k_{max}}{\min}\Gamma_{RRT}^{\alpha}(k)}$ in Theorem \ref{thm:rrm_finite} is lower than the $\epsilon_{rrt}=\dfrac{\Gamma_{RRT}^{\alpha}(k_0)\sqrt{1-\delta_{{k_0}}}\boldsymbol{\beta}_{min}}
{1+\Gamma_{RRT}^{\alpha}(k_0)}$ of Lemma \ref{lemma:icml} for RRT. Further, unlike RRT where the support recovery guarantee depends only on $\boldsymbol{\beta}_{min}$, the support recovery guarantee for RRM involves the term $DR(\boldsymbol{\beta})=\boldsymbol{\beta}_{max}/\boldsymbol{\beta}_{min}$. In particular, the SNR required for exact support recovery using RRM increases with increasing  $DR(\boldsymbol{\beta})$. This limits the finite SNR utility of RRM for signals with high $DR(\boldsymbol{\beta})$.  
\end{remark}
\begin{remark} The detiorating performance of RRM with increasing $DR(\boldsymbol{\beta})$ can be explained as follows. Note that both RRM and RRT try to identify the sudden decrease in $\|{\bf r}^k\|_2$ compared to $\|{\bf r}^{k-1}\|_2$ once $\mathcal{S}\subseteq \mathcal{S}_k$ for the first time, i.e., when $k=k_{min}$. This sudden decrease in $\|{\bf r}^k\|_2$ at $k=k_{min}$ is due to the removal of signal component in ${\bf r}^k$ at the $k_{min}^{th}$ iteration. However, a similar dip in $\|{\bf r}^k\|_2$ compared to $\|{\bf r}^{k-1}\|_2$ can also happens at a $k<k_{min}^{th}$ iteration if $\boldsymbol{\beta}_{t^k}$ ($t^k$ is the index selected by OMP in it's $k^{th}$ iteration) contains  most of the energy in the regression vector $\boldsymbol{\beta}$.  This  intermediate dip  in $RR(k)$ can be more pronounced than the dip happening  at the $k_{min}^{th}$ iteration when the SNR is moderate and $DR(\boldsymbol{\beta})$ is high. Consequently, the RRM estimate has a tendency to underestimate $k_{min}$ (i.e., $k_{RRM}<k_{min}$)  when $DR(\boldsymbol{\beta})$ is high. Note that by Lemmas \ref{lemma:latest_omp} and \ref{lemma:RR_properties}, $k_{min}=k_0$ and $\mathcal{S}_{k_{min}}=\mathcal{S}$ with a very high probability once $\epsilon_{\sigma}<\epsilon_{omp}$. Hence, this tendency of RRM to underestimate $k_{min}$ results in a support recovery error when $DR(\boldsymbol{\beta})$ is high.  This  behaviour of RRM is reflected in the higher SNR required for  recovering the support of $\boldsymbol{\beta}$ with higher $DR(\boldsymbol{\beta})$. Please note that RRM will be consistent at high SNR irrespective of the value of $DR(\boldsymbol{\beta})$. Since RRT is looking for the ``last" significant dip instead of the ``most significant" dip in $RR(k)$, RRT is not affected by the variations in $DR(\boldsymbol{\beta})$.

\end{remark}
\subsection{Residual ratio thresholding with adaptation (RRTA)}
As aforementioned, inspite of it's HSC and noise statistics oblivious nature, RRM has poor finite SNR support recovery performance for signals with high dynamic range and good high SNR performance for all types of signals. In contrast to RRM,  RRT has good finite SNR performance and inferior high SNR performance. This motivates the  RRTA algorithm which tries to combine the strengths of both RRT and RRM to produce a high SNR consistent support recovery scheme that also has good finite SNR performance. Recall from Lemma \ref{lemma:RR_properties} that when the matrix ${\bf X}$ satisfies the RIC condition $\delta_{k_0+1}<\dfrac{1}{\sqrt{k_0+1}}$, then the minimum value of $RR(k)$ decreases to zero with increasing SNR. Also recall that the reason for the high SNR inconsistency of RRT lies in our inability to adapt the RRT hyperparameter $\alpha$  with respect  to the operating SNR. In particular, for the high SNR consistency of RRT, we need  to enable the adaptation $\alpha\rightarrow 0$ as $\sigma^2\rightarrow 0$ without knowing $\sigma^2$.  Even though $\sigma^2$ is a  parameter very difficult to estimate, given that the minimum value of $RR(k)$ decreases to zero with increasing SNR or decreasing $\sigma^2$,  it is still possible to adapt $\alpha\rightarrow 0$ with increasing SNR by making $\alpha$ a monotonically increasing function of $\underset{k}{\min}RR(k)$. This is the proposed  RRTA technique which ``adapts" the $\alpha$ parameter in the RRT algorithm using $\underset{k}{\min}RR(k)$.  The support estimate in RRTA can be formally expressed as 
\begin{equation}
 \mathcal{S}_{RRTA}=\mathcal{S}_{k_{RRTA}},\ \text{where} \ k_{RRTA}=\max\{k:RR(k)<\Gamma_{RRT}^{\alpha*}(k)\}\  
\end{equation}   
\begin{equation}
\text{and} \ \ \alpha^*=\min\left(PFD_{finite},\underset{k}{\min}RR(k)^q\right) 
\end{equation}
for some $q>0$ and $PFD_{finite}>0$. RRTA algorithm has two user defined parameters,i.e., $PFD_{finite}$ and $q$ which control the finite SNR and high SNR behaviours respectively.  

We first discuss the choice of hyperparameter $PFD_{finite}$. Note that $\underset{k}{\min}RR(k)^q$ 
will take  small values only at high SNR. Hence, with a choice of  $\alpha^*=\min\left(PFD_{finite},\underset{k}{\min}RR(k)^q\right)$,  RRTA will operate like RRT with $\alpha=PFD_{finite}$ in the low to moderate high SNR regime, whereas, RRTA will operate like RRT  with $\alpha=\underset{k}{\min}RR(k)^q$ in the high SNR regime. As discussed in Lemma \ref{lemma:icml}, $RRT$ with a date independent parameter $\alpha$ can control the probability of false discovery (PFD) in the finite SNR regime within $\alpha$. Hence, the user specified value of $PFD_{finite}$ specifies the  maximum finite SNR probability of false discovery allowed in RRTA. This is a design choice. Following the empirical results and recommendations in  \cite{icml} we set this parameter to $PFD_{finite}=0.1$.   

As aforementioned, the choice of second hyperparameter $q$  determines the high SNR behaviour of RRTA. The following theorem specifies the requirements on the hyperparameter $q$ such that RRTA is a high SNR consistent estimator of the true support $\mathcal{S}$.
\begin{thm}\label{thm:RRTA}
Suppose that the matrix satisfies RIP condition of order $k_0+1$ and $\delta_{k_0+1}<1/\sqrt{k_0+1}$. Also suppose that $\alpha^{*}=\min\left(PFD_{finite}, \underset{k}{\arg\min}RR(k)^q\right)$ for some $q>0$. Then RRTA is high SNR consistent, i.e., $\underset{\sigma^2 \rightarrow 0}{\lim}\mathbb{P}(\mathcal{S}_{RRTA}=\mathcal{S})=1$ for any fixed $PFD_{finite}>0$ once $n>k_0+q$. 
\end{thm}
\begin{proof}
Please see Appendix D.
\end{proof}
\begin{remark} Note that $\alpha^*=\min\left(PFD_{finite}, \underset{k}{\arg\min}RR(k)^q\right)$ is a monotonic function of $\underset{k}{\arg\min}RR(k)$ at high SNR for all values of $q>0$. However,  the rate at which  $\underset{k}{\arg\min}RR(k)^q$ decreases to zero increases with increasing values of $q$. The constriants $q>0$ and $q<n-k_0$ ensures that $\underset{k}{\arg\min}RR(k)^q$ should decrease to zero at a rate that is not too high. A very small value of $q$ implies that $\underset{k}{\arg\min}RR(k)^q$ will be greater than $PFD_{finite}=0.1$ for the entire operating SNR range thereby denying RRTA the required SNR adaptation, whereas, a very large value of $q$ implies that $\underset{k}{\arg\min}RR(k)^q<0.1$ even for low SNR. Operating RRT  with a very low value of  $\alpha$ at low SNR results in inferior performance. Hence, the choice of $q$ is important in RRTA. Through extensive numerical simulations, we observed that a value of $q=2$ delivers the best overall performance in the low and high SNR regimes. Such subjective choices are also involved in the noise variance aware HSC results developed in \cite{elsevier}.   
\end{remark}
\begin{remark} With the choice of $q=2$, the constraint $n>k_0+2$ has to be satisifed for the HSC of RRTA. Note that $k_0$ is unknown \textit{a priori} and hence it is impossible to check the condition $n>k_0+2$.  However,  successfull sparse recovery using any sparse recovery technique requires $k_0<\floor{\frac{n+1}{2}}$ or equivalently $n>2k_0-1$. Note that $2k_0-1>k_0+2$ for all $k_0>3$. In addition to this,   the regularity condition $\delta_{k_0+1}<\dfrac{1}{\sqrt{k_0+1}}$  required for sparse recovery using OMP will be satisfied in many widely used matrices ${\bf X}$ if $n=O\left(k_0^2\log(p)\right)$. Hence, the condition $n>k_0+2$ will be  satisfied automatically in all problems where OMP is expected to carry out successfull sparse recovery.   
\end{remark}
\begin{remark}Note that the poor performance of RRM with increasing $DR(\boldsymbol{\beta})$ is pronounced in the low to moderate SNR regime.  Note that with $q=2$, RRTA works exactly like RRT with $\alpha=PFD_{finite}=0.1$ in the low to moderate SNR regime.   Since, RRT is not affected by the value of  $DR(\boldsymbol{\beta})$, RRTA also will not be affected by high $DR(\boldsymbol{\beta})$. 
\end{remark}
\begin{remark}Note that we are setting  $\alpha^*=\min\left(PFD_{finite},\underset{k}{\min}RR(k)^q\right)$ instead of  $\alpha^*=\underset{k}{\min}RR(k)^q$. This is to ensure that  RRTA  work as RRT with parameter $\alpha^*=PFD_{finite}$  when the SNR is in the small to moderate regime. In other words, this particular form of $\alpha^*$ will help keep the values of $\alpha$ motivated by HSC arguments applicable only at high SNR and values of $\alpha$ motivated by finite SNR performance applicable to finite SNR situations. 

\end{remark}
Given the stochastic nature of the hyperparameter $\alpha^*$ in RRTA, it is difficult to derive finite SNR guarantees for RRTA. However, numerical simulations given in Section \rom{4} indicate that RRTA delivers a performance very close to that of RRT with parameter $\alpha=0.1$ and OMP with \textit{a priori} knowledge of $k_0$ or $\sigma^2$ in the low to moderately  high SNR regime.

\subsection{RRM and RRTA algorithms: A discussion}
In this section, we compare and contrast  RRM and RRTA  algorithms  proposed in this article with the results and algorithms discussed in existing literature.  These comparisons are organized as seperate remarks.  We first compare the RRTA and RRM algorithms with exisiting algorithms.
\begin{remark}
Note that  RRT  has a user specified parameter $\alpha$ which was set to $\alpha=0.1$ in \cite{icml} based on finite sample estimation and support recovery performance. This choice of $\alpha=0.1$ is also carried over to RRTA which uses $\alpha^*=\min\left(PFD_{finite},\underset{k}{\min}RR(k)^q \right)$ and $PFD_{finite}=0.1$. In addition to this, RRTA also involves a hyperparameter $q$ which is also set based on analytical results and empirical performance.   In contrast, RRM does not involve any user specified parameter. Hence, while RRT and RRTA are signal and noise statistics oblivious  algorithms, i.e., algorithms that does not require \textit{a priori} knowledge of $k_0$ or $\sigma^2$ for efficient finite or high SNR  operation, RRM is  both signal and noise statistics oblivious and hyper parameter free. Recently, there has been a significant interest in the development of such hyper parameter free algorithms. Significant developments in this area of research include algorithms related to sparse inverse covariance estimator, \textit{aka} SPICE  \cite{spicenote,spice_connection,spice,spice_like,Stoica20141}, sparse Bayesian learning (SBL)\cite{wipf2004sparse} etc. However, to the best of our knowledge, no explict finite sample and finite SNR support recovery guarantees are developed for SPICE, it's variants or SBL. The HSC of these algorithms are also not discussed in literature. In contrast, RRM is a hyper parameter free algorithm which is  high SNR consistent.  Further, RRM has explicit finite sample and finite SNR support recovery guarantees. Also, for signals with low dynamic range, the performance of RRM is shown analytically to be comparable with OMP having \textit{a priori} knowledge of $k_0$ or $\sigma^2$. 
\end{remark}
\begin{remark}
As aforementioned, all previous literature regarding  HSC in low dimensional and high dimensional regression models were applicable only to situations where the noise variance $\sigma^2$ is known \textit{a priori} or easily estimable.   In contrast, RRM and RRTA can achieve high SNR consistency   in underdetermined regression models even  without requiring an estimate of $\sigma^2$. To the best of our knowledge, RRM and RRTA  are the first and only noise statistics oblivious algorithms that are shown to be high SNR consistent in underdetermined regression models. 
\end{remark}
\begin{remark}
The HSC results developed in this article rely heavily on the bound $\mathbb{P}\left(RR(k)>\Gamma_{RRT}^{\alpha}(k),\forall k>k_{min}\right)\geq 1-\alpha$ in Lemma \ref{lemma:RR_properties}.  This bound is valid  iff the support sequence $\mathcal{S}_k$ is monotonic, i.e., $\mathcal{S}_k\subset\mathcal{S}_{k+1}$.  Unfortunately, the support sequences produced by sparse recovery algorithms like LASSO, SP, CoSaMP etc. are not monotonic.  Hence, the RRT technique in \cite{icml} and the RRM/RRTA techniques proposed in this article  are not applicable to non monotonic algorithms like LASSO, SP, CoSaMP etc. Developing versions of  RRT, RRM and RRTA that are applicable  to non monotonic algorithms like LASSO, CoSaMP, SP etc. is an important direction for future research. 
\end{remark}
Next we discuss the matrix regularity conditions involved in deriving RRTA and RRM algorothms. 
\begin{remark}\label{rem:mic}
Please note that the HSC of RRM and RRTA are derived  assuming only the existence of a matrix regularity condition which ensures that  $\mathcal{S}_{k_0}=\mathcal{S}$ once $\|{\bf w}\|_2\leq \epsilon$ for some $\epsilon>0$. RIC based regularity conditions are used in this article because they are the most widely used in analysing OMP. However, two other regularity conditions, viz., mutual incoherence condition  (MIC) and exact recovery condition (ERC) are also used for analysing OMP.  
  The mutual incoherence condition $\mu_{\bf X}=\underset{j\neq k}{\max}|{\bf X}_j^T{\bf X}_k| <\dfrac{1}{2k_0-1}$  along with  $\|{\bf w}\|_2\leq \dfrac{\boldsymbol{\beta}_{min}(1-(2k_0-1)\mu_{\bf X})}{2} $  implies that $\mathcal{S}_{k_0}=\mathcal{S}$. Similarly, the exact recovery condition (ERC) $\underset{j \notin \mathcal{S}}{\max}\|{\bf X}_{\mathcal{S}}^{\dagger}{\bf X}_j\|_1<1$ along with $\|{\bf w}\|_2\leq \dfrac{\boldsymbol{\beta}_{min}\lambda_{min}(1-\|{\bf X}_{\mathcal{S}}^{\dagger}{\bf X}_j\|_1)}{2}$
also ensures that $\mathcal{S}_{k_0}=\mathcal{S}$ \cite{cai2011orthogonal}. Consequently, both RRTA and RRM are high SNR consistent once MIC or ERC are satisfied. 
 \end{remark}
 \begin{figure*}
 \begin{multicols}{2}
    
    \includegraphics[width=1\linewidth]{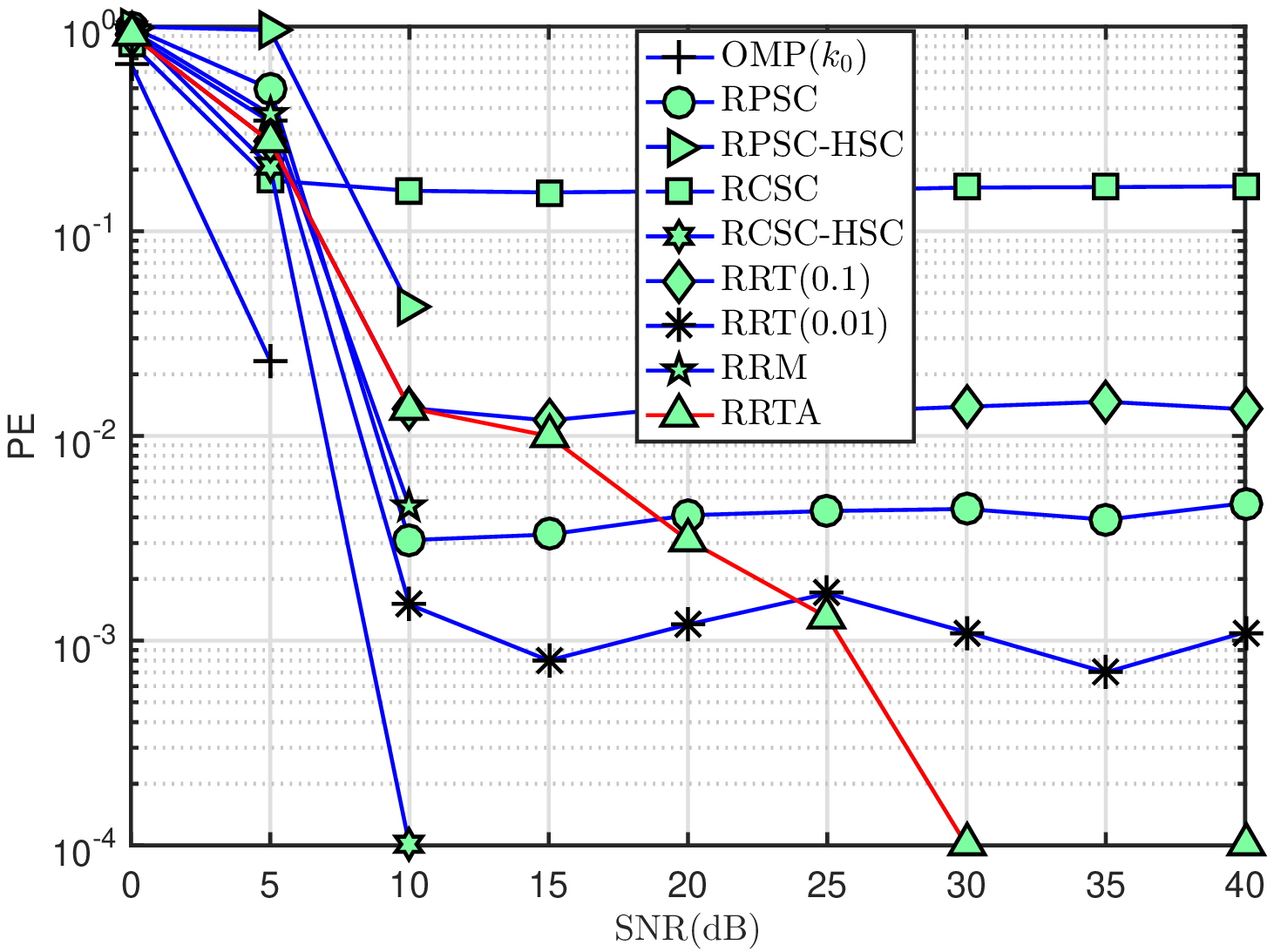}
    \caption*{ ${\bf X}=[{\bf I}_{32},{\bf H}_{32}]$, $\boldsymbol{\beta}_j=\pm1$.} 
    
    \includegraphics[width=1\linewidth]{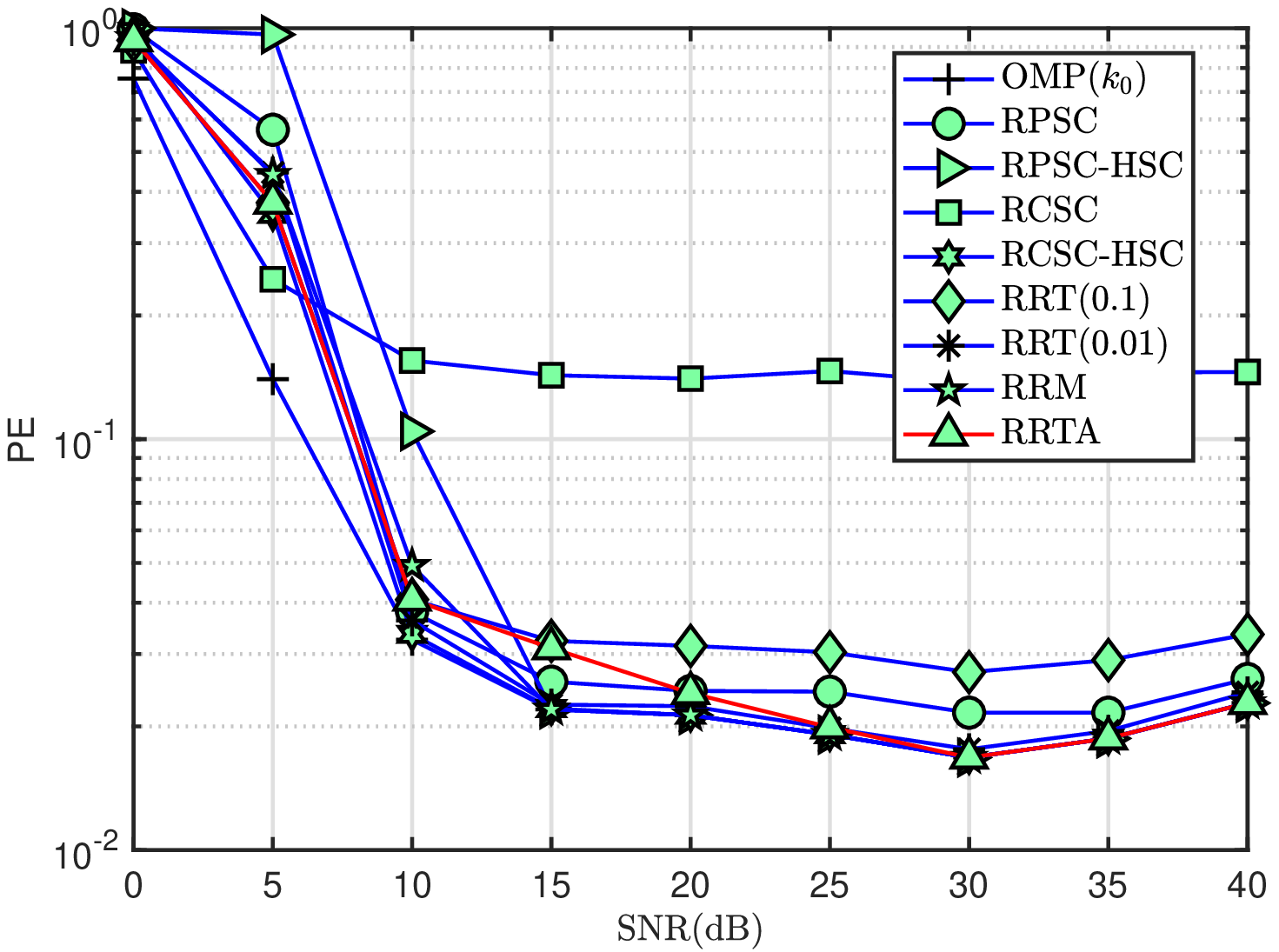} 
    \caption*{ ${\bf X}_{i,j}\sim \mathcal{N}(0,1/n)$, $\boldsymbol{\beta}_j=\pm1$.} 
    
      \end{multicols}
         \caption{Performance of RRM/RRTA when $DR(\boldsymbol{\beta})$ is low. }
         \label{fig:pm1}
   
   \end{figure*}

\begin{remark}\label{rem:random} When the matrix is generated by randomly sampling  from a $\mathcal{N}(0,1)$ distribution, then for every $\delta\in (0,0.36)$, $n\geq ck_0\log\left(\dfrac{p}{\delta}\right)$ where $c\leq 20$ is a  constant ensures that $\mathcal{S}_{k_0}=\mathcal{S}$ with a probability greater than $1-2\delta$ (when ${\bf w}={\bf 0}_n$) \cite{tropp2007signal}. Hence, even in the absence of noise ${\bf w}$, there is a fixed probability $\approx 2\delta$ that   $\mathcal{S}_{k_0}\neq \mathcal{S}$. This result implies that even OMP running exactly $k_0$ iterations cannot recover the true support with arbitrary high probability as $\sigma^2\rightarrow 0$ in such situations. Consequently, no OMP based scheme can be HSC when the matrix is randomly generated. However, numerical simulations indicate that the PE of RRM/RRTA and OMP with \textit{a priori} knowledge of $k_0$ match at high SNR, i.e., RRM  achieves the best possible performance that can be delivered by OMP. 
\end{remark}

\section{Numerical Simulations}

In this section, we numerically verify the HSC results derived for RRM and RRTA. We also evaluate and compare the finite SNR performance of RRM and RRTA with respect to other popular OMP based support recovery schemes.  We consider two models for the design matrix ${\bf X}$. Model 1 has ${\bf X}=[{\bf I}_n,{\bf H}_n]\in \mathbb{R}^{n \times 2n}$, i.e., ${\bf X}$ is formed by the concatenation of a $n\times n$ identity matrix and a  $n \times n$ Hadamard matrix. This matrix has $\mu_{\bf X}=\dfrac{1}{\sqrt{n}}$\cite{elad_book}. Consequently,  this matrix satisfies the MIC  $\mu_{\bf X}<\dfrac{1}{2k_0-1}$  for all vectors $\boldsymbol{\beta}\in \mathbb{R}^{2n}$
 with sparsity $k_0\leq \floor{\dfrac{1+\sqrt{n}}{2}}$. Model 2 is a $n \times p$ random matrix formed by sampling the entries independently from a $\mathcal{N}(0,1/n)$ distribution. For a given sparsity $k_0$, this matrix satisfies $\mathcal{S}_{k_0}=\mathcal{S}$ at high SNR with a high probability whenever $k_0 =O(\frac{n}{\log(p)})$.  Please note that there is a nonzero probability that a random matrix fails  to satisfy the regularity conditions required for support recovery. Hence, no OMP scheme is expected to be high SNR consistent in matrices of model 2. In both cases the dimensions are set as $n=32$ and $p=2n=64$. 
 
 \begin{figure*} 
    \begin{multicols}{2}
    \includegraphics[width=1\linewidth]{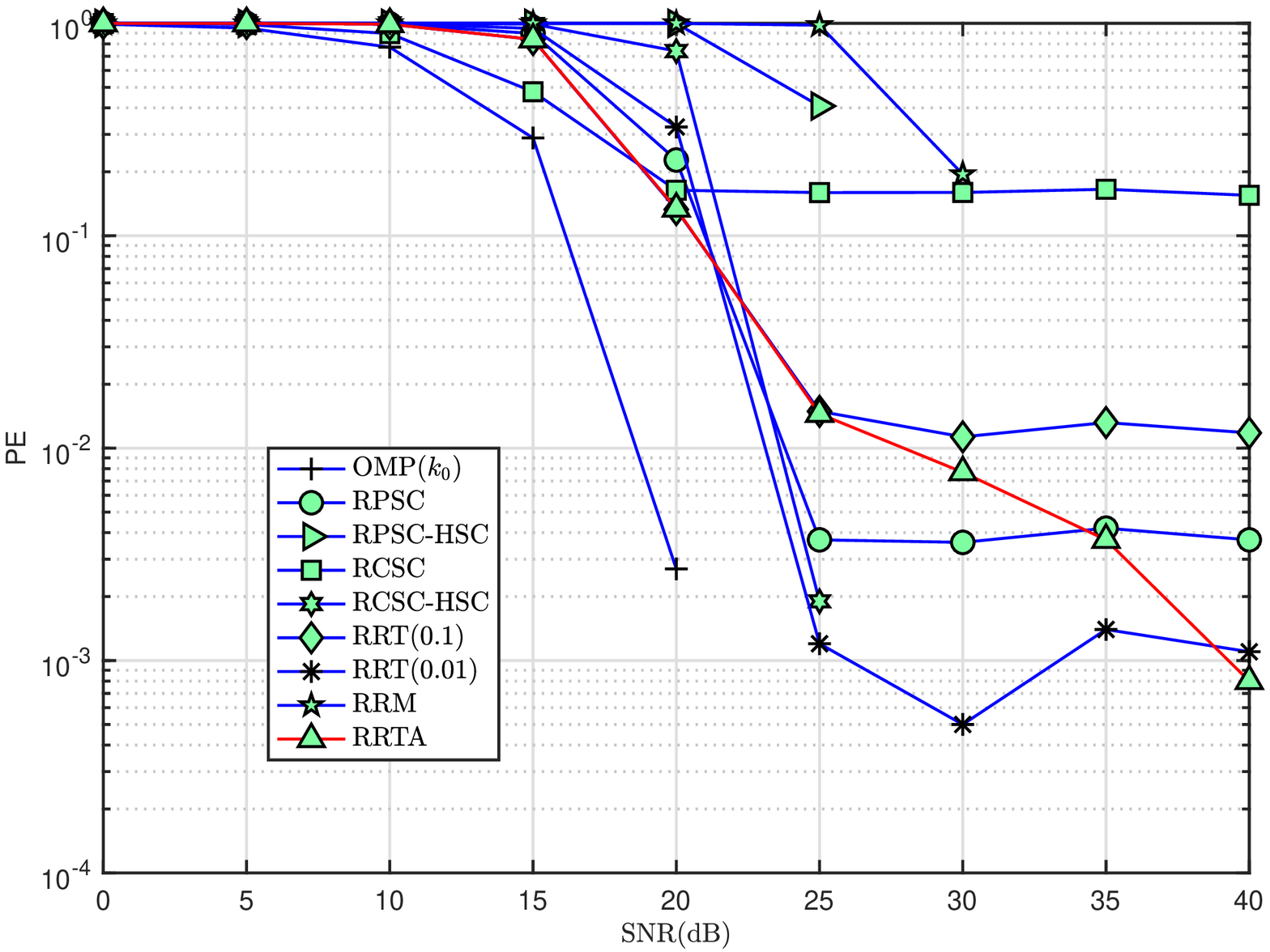}
    \caption*{ ${\bf X}=[{\bf I}_{32},{\bf H}_{32}]$, $\boldsymbol{\beta}_j\in\{1,1/3,1/9\}$.} 
    
    \includegraphics[width=1\linewidth]{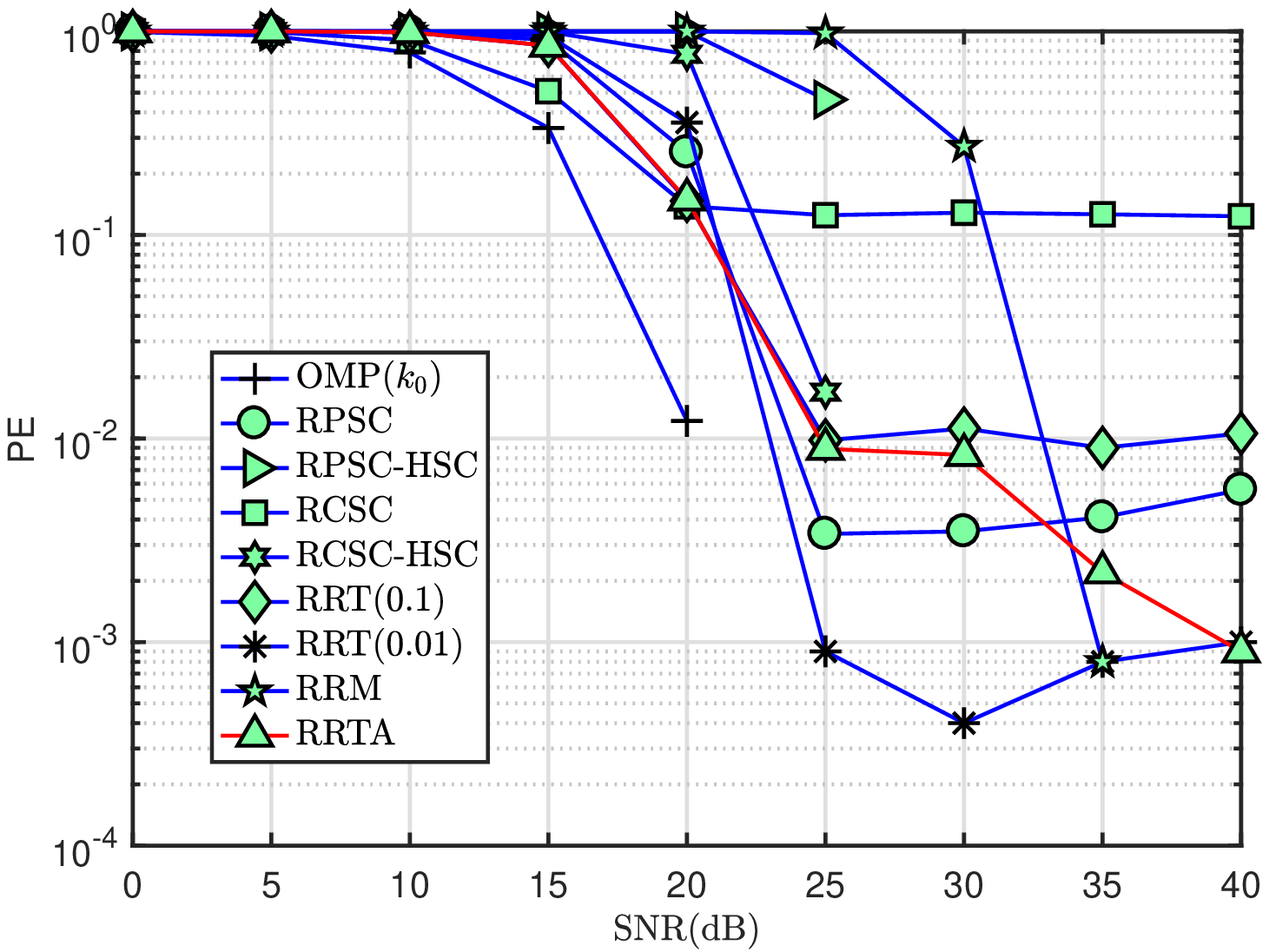} 
    \caption*{ ${\bf X}_{i,j}\sim \mathcal{N}(0,1/n)$, $\boldsymbol{\beta}_j \in \{1,1/3,1/9\}$.} 
  
   \end{multicols}
   \caption{Performance of RRM/RRTA when $DR(\boldsymbol{\beta})$ is high. }
   \label{fig:decay}
\end{figure*} 
\begin{figure*} 
    \begin{multicols}{2}
    
    \includegraphics[width=1\linewidth]{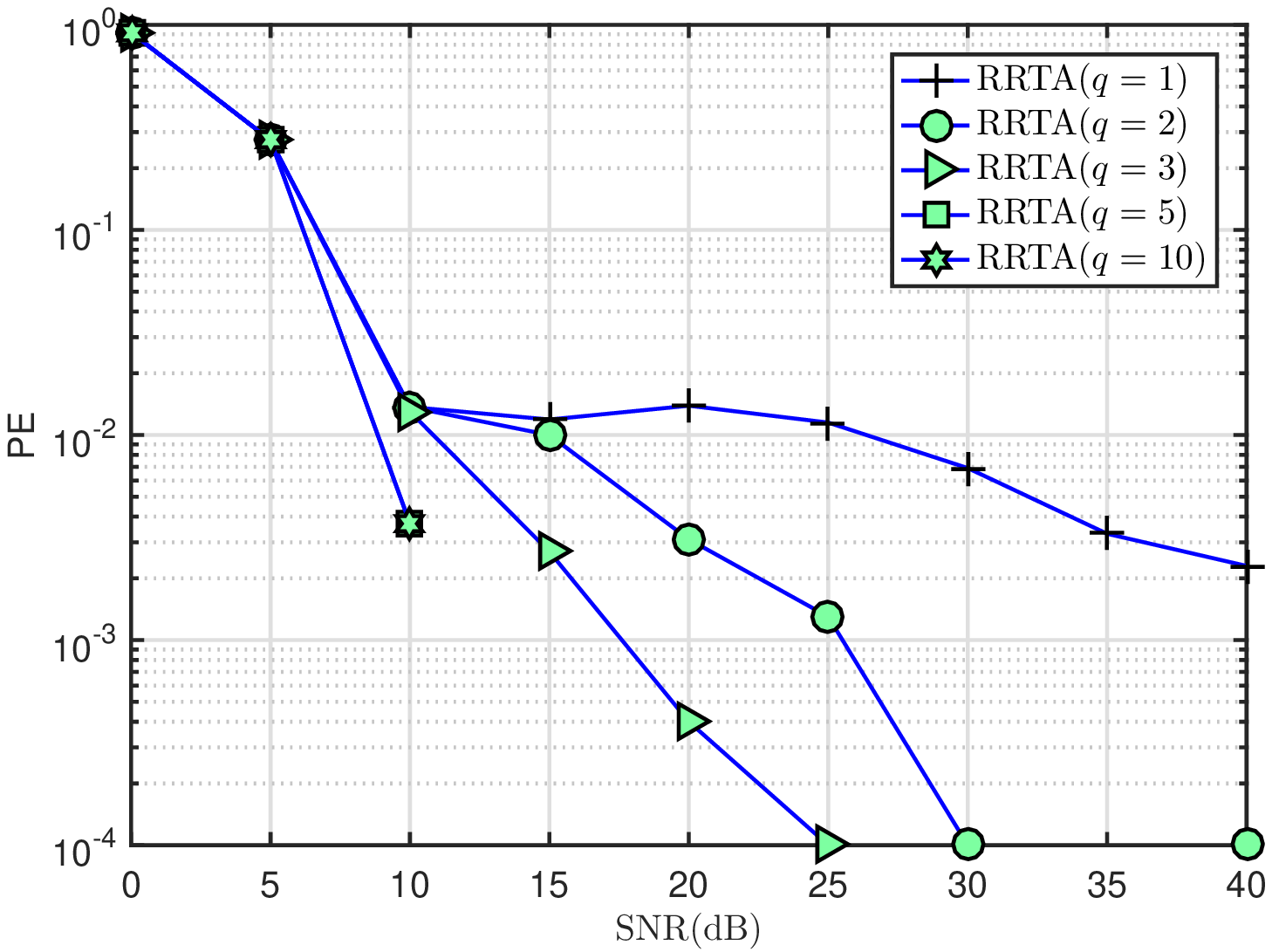}
    \caption*{ ${\bf X}=[{\bf I}_{32},{\bf H}_{32}]$, $\boldsymbol{\beta}_j=\pm 1$.} 
    
     \includegraphics[width=1\linewidth]{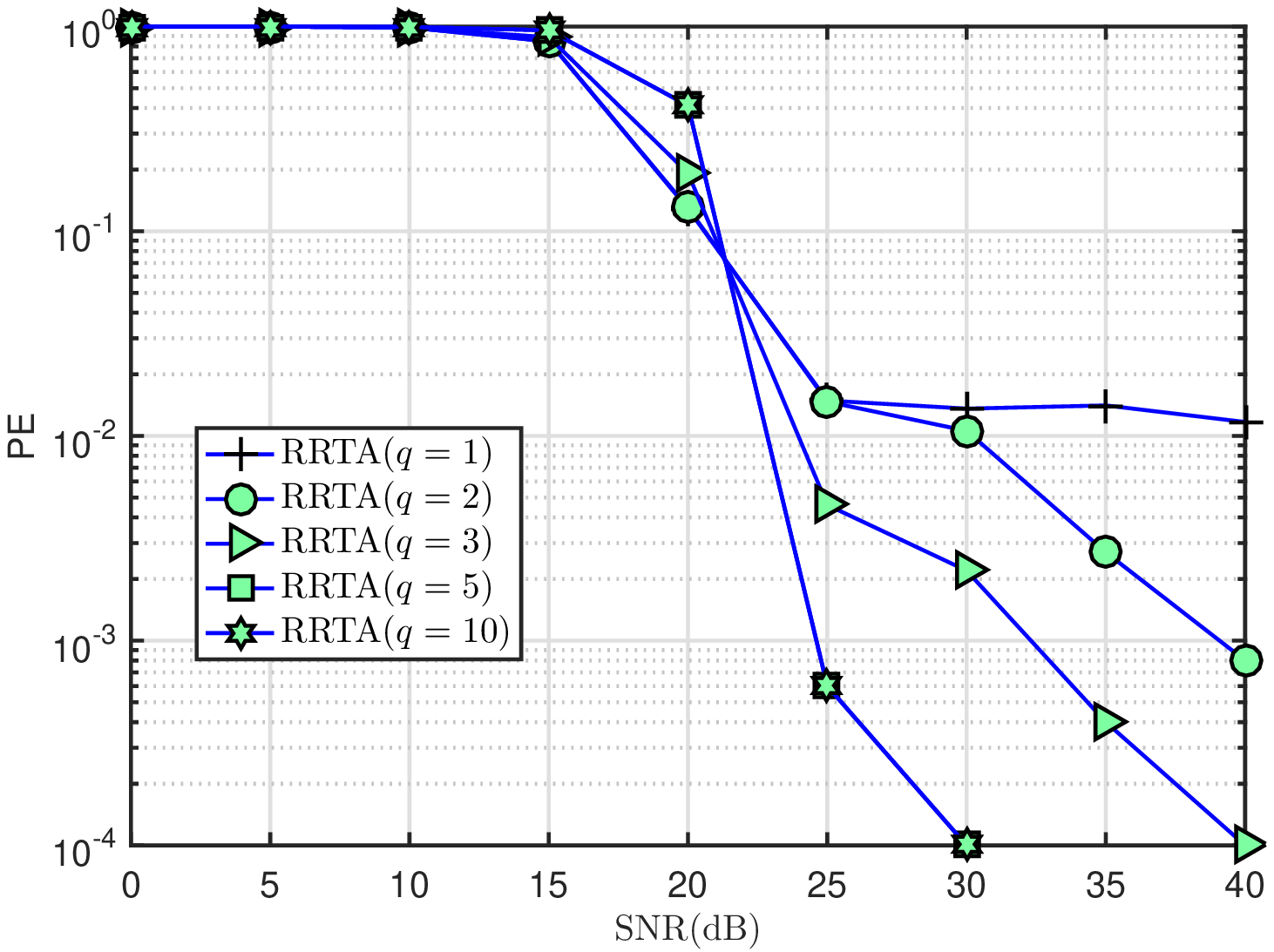} 
     \caption*{ ${\bf X}=[{\bf I}_{32},{\bf H}_{32}]$, $\boldsymbol{\beta}_j\in\{1,1/3,1/9\}$.} 
   
   \end{multicols}
   \caption{Effect of $q$ on the RRTA performance.}
   \label{fig:effectq}
\end{figure*}

%
%

Next we discuss   the regression vector $\boldsymbol{\beta}$ used in our experiments. The support $\mathcal{S}$  is  obtained by  randomly sampling $k_0$ entries  from the set $\{1,2,\dotsc,p\}$. Sparsity level $k_0$  is fixed at $k_0=3$. Since the performance of RRM is influenced by $DR(\boldsymbol{\beta})$,  we consider two types of  regression vectors $\boldsymbol{\beta}$ in our experiments with significantly different values of  $DR(\boldsymbol{\beta})$. For the first type of $\boldsymbol{\beta}$, the non zero entries  are randomly assigned $\pm 1$. When $\boldsymbol{\beta}_j=\pm 1$ for $j\in \mathcal{S}$, $DR(\boldsymbol{\beta})$ is at it's lowest value, i.e.,  one.  For the second type of $\boldsymbol{\beta}$,  the non zero entries are sampled from the set $\{1,1/3,\dotsc,1/3^{k_0-1}\}$ without resampling.  Here the dynamic range $DR(\boldsymbol{\beta})=3^{k_0-1}=9$ is very high.    All results presented in this section are obtained after performing $10^4$ Monte carlo iterations. 
%

 \subsection{Algorithms under consideration}
  Figures \ref{fig:pm1}-\ref{fig:effectq} present the $PE$ versus $SNR$ curves for  the four possible combinations of design matrix/regression vectors discussed earlier. OMP($k_0$) in Fig \ref{fig:pm1}-\ref{fig:effectq} represent the OMP scheme that performs exactly $k_0$ iterations. OMP with residual power stopping condition (RPSC)  stops iterations once $\|{\bf r}^{(k)}\|_2\leq \sigma\Gamma_1$. RPSC in Fig.\ref{fig:pm1}-\ref{fig:decay} represent this scheme with $\Gamma_1=\sqrt{n+2\sqrt{n\log(n)}}$\cite{cai2011orthogonal}. OMP with residual correlation stopping condition (RCSC)  stops iterations $\|{\bf X}^T{\bf r}^{(k)}\|_{\infty}\leq \sigma\Gamma_2$. RCSC in Fig.\ref{fig:pm1}-\ref{fig:decay} represent  this scheme with $\Gamma_2=\sqrt{2\log(p)}$. 
 RPSC-HSC and RCSC-HSC represent RPSC and RCSC with $\Gamma_1=\dfrac{1}{\sigma^{\eta}}\sqrt{n+2\sqrt{n\log(n)}}$ and  $\Gamma_2=\dfrac{1}{\sigma^{\eta}}\sqrt{2\log(p)}$ respectively\cite{cai2011orthogonal}. By Lemma \ref{lemma:latest_omp},  both RPSC-HSC and RCSC-HSC are high SNR consistent once $0<\eta<1$\cite{elsevier}.  In our simulations, we set  $\eta= 0.1$ as suggested in \cite{elsevier}. RRT($\alpha$) in Fig.\ref{fig:pm1}-\ref{fig:decay} represent RRT with $\alpha=0.1$ and $\alpha=0.01$. RRTA in Fig.\ref{fig:pm1}-\ref{fig:decay} represent RRTA with $\alpha^*=\min(0.1,\underset{k}{\min}RR(k)^2)$, i.e., the parameters $PFD_{finite}$ and $q$ are set to $0.1$  and $2$ respectively. RRTA($q= $) in Fig.\ref{fig:effectq} represents RRTA with $\alpha^*=\min(0.1,\underset{k}{\min}RR(k)^q)$.

 \subsection{Comparison of RRM/RRTA with existing OMP based support recovery techniques}
 
  Fig.\ref{fig:pm1} presents  the PE performance of algorithms when $DR(\boldsymbol{\beta})$ is low. When the  condition $\mu_{\bf X}<\dfrac{1}{2k_0-1}$  is met,  one can see from Fig.\ref{fig:pm1} that the PE of RRM, RRTA, RCSC-HSC, RPSC-HSC and OMP($k_0$) decreases to zero with increasing SNR. This validates the claims made in Lemma \ref{lemma:latest_omp},   Theorem \ref{thm:rmm_hsc} and Theorem \ref{thm:RRTA}.   Please note that unlike OMP($k_0$) which has \textit{a priori} knowledge of $k_0$ and RCSC-HSC and RPSC-HSC with \textit{a priori} knowledge of $\sigma^2$, RRM  and RRTA achieve HSC without having \textit{a priori} knowledge of either signal or noise statistics.   OMP($k_0$) achieves the best PE performance, whereas, the performance of other schmes are comparable to each other in the low to moderate SNR. This validates the claim made in Theorem \ref{thm:rrm_finite} that RRM performs similar to the noise statistics aware OMP schemes  when $DR(\boldsymbol{\beta})$ is low.  However, at high SNR, the rate at which the PE of RRTA converges to zero is lower than the rate at which the PE of RRM, RPSC-HSC etc. decrease to zero.   PE of RPSC, RCSC  and RRT are close to OMP($k_0$) at low SNR. However, the PE versus SNR curves of these algorithms  exhibit a tendency to  floor with increasing  SNR resulting in high SNR inconsistency.  When the design matrix ${\bf X}$  is randomly generated, no OMP based scheme achieves HSC. However, the PE level at which RRTA, RRM etc. floor is same as the PE level at which  OMP($k_0$), RPSC-HSC and RCSC-HSC floor. In other words,  when HSC is not achievable, RRTA and RRM will deliver a high SNR PE performance similar to the signal and noise statistics aware OMP schemes. 
  
  Next we consider the performance of algorithms when $DR(\boldsymbol{\beta})$ is high. Comparing Fig.\ref{fig:pm1} and Fig.\ref{fig:decay}, one can see that the $PE$ versus SNR curves for all algorithms shift towards the high SNR region with increasing $DR(\boldsymbol{\beta})$. Note that for a fixed SNR, $\boldsymbol{\beta}_{min}/\sigma$ decreases with increasing $DR(\boldsymbol{\beta})$. Since all OMP based schemes require $\boldsymbol{\beta}_{min}/\sigma$ to be sufficiently high, the relatively poor performance with increasing $DR(\boldsymbol{\beta})$ is expected. However, as one can see from Fig.\ref{fig:decay}, the deterioration in performance with increasing $DR(\boldsymbol{\beta})$ is very severe in RRM compared to other OMP based algorithms. This verifies the finite sample results derived in Theorem \ref{thm:rrm_finite} for RRM which states that RRM has poor finite SNR performance when $DR(\boldsymbol{\beta})$ is high. Note that the performance of RRTA with increased $DR(\boldsymbol{\beta})$ is similar to that of  OMP($k_0$), RPSC, RCSC etc. in the finite SNR regime. This implies that unlike RRM, the performance of RRTA depends only on $\boldsymbol{\beta}_{min}$ and not $DR(\boldsymbol{\beta})$.  
 \subsection{Effect of parameter $q$ on RRTA performance}

 Theorem \ref{thm:RRTA} states that RRTA with all values of $q$ satisfying $0<q<n-k_0$ are high SNR consistent. However, the finite SNR performance of RRTA with different values of  $q$ will be different. In Fig.\ref{fig:effectq}, we evaluate the performance of RRTA for different values of $q$. As one can see from Fig.\ref{fig:effectq}, the rate at which PE decreases to zero with increasing SNR becomes faster  with the increase in $q$. The rate at which the PE of RRTA with $q=10$ decreases to zero is similar to the steep decrease seen in the PE versus SNR curves of signal and noise statistics aware schemes like OMP($k_0$), RPSC-HSC, RCSC-HSC etc. In contrast, the rate at which the PE of RRTA with $q=1$  decreases to zero is not  steep. RRTA with $q=2$ exhibit  a much steeper  PE versus SNR curve. However, as one can see from the R.H.S of Fig.\ref{fig:effectq}, the finite SNR performance of RRTA with $q=1$ and $q=2$ is better than that of RRTA with $q=5$, $q=10$ etc. In other words, RRTA with a larger value of $q$ can potentially yield a better PE than RRTA with a smaller value of $q$ in the high SNR regime. However, this come at the cost of an  increased  PE in the  low to medium SNR regime.  Since the objective of any good HSC scheme should be to achieve HSC while guaranteeing good finite SNR performance, one can argue that $q=2$ is a good design  choice for the hyperparameter in RRTA.

 \section{Conclusion}
 This article proposes two novel techniques called RRM and RRTA to operate OMP without the \textit{a priori} knowledge of signal sparsity or noise variance. RRM is  hyperparameter free in the sense that it does not require any user specified tuning parameters, whereas, RRTA involve hyperparameters.  We analytically establish the HSC of both RRM and RRTA. Further, we also derive finite SNR guarantees for RRM. { Numerical simulations also verify the HSC of RRM and RRTA. RRM and RRTA are the first signal and noise statistics oblivious techniques to report HSC in underdetermined regression models. }
 \section*{Appendix A: Proof of Lemma \ref{lemma:inconsistency} }
 \begin{proof}
 The event $\mathcal{S}_{RRT}\supset \mathcal{S}$ in RRT estimate $k_{RRT}=\max\{k:RR(k)<\Gamma_{RRT}^{\alpha}(k)\}$ is true once $\exists k>k_{min}$ such that $RR(k)<\Gamma_{RRT}^{\alpha}(k)$. Hence, 
 \begin{equation}\label{eq:inconsistency1}
 \mathbb{P}(\mathcal{S}_{RRT}\supset \mathcal{S}) \geq  \mathbb{P}\left(\underset{k>k_{min}}{\bigcup}\{RR(k)<\Gamma_{RRT}^{\alpha}(k)\}\right).
 \end{equation}
 One can further bound (\ref{eq:inconsistency1}) as follows. 
 \begin{equation}\label{eq:inconsistency2}
 \begin{array}{ll}
 \mathbb{P}(\mathcal{S}_{RRT}\supset \mathcal{S}) \overset{(a)}{\geq}  \mathbb{P}\big(\{RR(k_{min}+1)<\Gamma_{RRT}^{\alpha}(k_{min}+1)\} \big) \\ \overset{(b)}{\geq} \mathbb{P}\big(\{RR(k_{min}+1)<\Gamma_{RRT}^{\alpha}(k_{min}+1)\}\cap \{k_{min}=k_0\}\big) \\
 = \mathbb{P}\left(\{RR(k_{min}+1)<\Gamma_{RRT}^{\alpha}(k_{min}+1)\} | \{k_{min}=k_0\}\right) \\
 \ \ \ \ \ \ \ \ \ \ \ \ \ \ \ \ \ \ \ \ \ \ \mathbb{P}(k_{min}=k_0).
 \end{array}
 \end{equation}
 { (a) of \ref{eq:inconsistency2} follows from the union bound $\mathbb{P}(A \cup B)\geq \mathbb{P}(A)$ and (b) follows from the intersection bound $\mathbb{P}(A\cap B)\leq \mathbb{P}(A)$.}
 Following the proof of Theorem 1 in \cite{icml}, we know that conditional on $k_{min}=j$, for each $k>j$, $RR(k)<Z_k$ where  $Z_k \sim \mathbb{B}(\frac{n-k}{2},\frac{1}{2})$. Applying this distributional result in $\mathbb{P}\left(\{RR(k_{min}+1)<\Gamma_{RRT}^{\alpha}(k_{min}+1)\}\right)$  gives
 \begin{equation}\label{eq:inconsistency3}
 \begin{array}{ll}
\mathbb{P}\left(\{RR(k_{min}+1)<\Gamma_{RRT}^{\alpha}(k_{min}+1)\}\right) \\ 
\  \ \ \ \ \ \ \ \ \ \ \ \ \ \geq \mathbb{P}\big(\{Z_{k_0+1} < \Gamma_{RRT}^{\alpha}(k_0+1)\}\big)\\
\  \ \ \ \ \ \ \ \ \ \ \ \ \ =F_{\frac{n-k_0-1}{2},\frac{1}{2}}\left(\Gamma_{RRT}^{\alpha}(k_0+1)\right)\\
\  \ \ \ \ \ \ \ \ \ \ \ \ \ =F_{\frac{n-k_0-1}{2},\frac{1}{2}}\left(F^{-1}_{\frac{n-k_0-1}{2},\frac{1}{2}}\left(\frac{\alpha}{k_{max}(p-k_0)}\right)\right)\\
\  \ \ \ \ \ \ \ \ \ \ \ \ \ =\frac{\alpha}{k_{max}(p-k_0)}  
\end{array}
 \end{equation}
Using Lemma \ref{lemma:latest_omp}, we know that $k_{min}=k_0$ once $\|{\bf w}\|_2\leq\epsilon_{omp}$. Hence, $\mathbb{P}(k_{min}=k_0)\geq \mathbb{P}(\|{\bf w}\|_2\leq \epsilon_{omp})$. Since ${\bf w} \overset{P}{\rightarrow} {\bf 0}_n$ as $\sigma^2\rightarrow 0$ for ${\bf w}\sim \mathcal{N}({\bf 0}_n,\sigma^2{\bf I}_n)$, we have  $\underset{\sigma^2\rightarrow 0}{\lim}\mathbb{P}(\|{\bf w}\|_2\leq \epsilon_{omp})=1$ and $\underset{\sigma^2\rightarrow 0}{\lim}\mathbb{P}(k_{min}=k_0)=1$. Substituting $\underset{\sigma^2\rightarrow 0}{\lim}\mathbb{P}(k_{min}=k_0)=1$ and (\ref{eq:inconsistency3}) in (\ref{eq:inconsistency2}) gives $\underset{\sigma^2 \rightarrow 0}{\lim}\mathbb{P}(\mathcal{S}_{RRT}\supset \mathcal{S})\geq \frac{\alpha}{k_{max}(p-k_0)}$. 
 \end{proof}
 \section*{Appendix B: Proof of Lemma \ref{lemma:kless}}
 \begin{proof}
 By Lemma \ref{lemma:latest_omp}, we have $k_{min}=k_0$ and $\mathcal{S}_{k_0}=\mathcal{S}$ once $\|{\bf w}\|_2\leq \epsilon_{omp}$.   This along with the monotonicity of $\mathcal{S}_k$ implies that $\mathcal{S}_k\subset  \mathcal{S}$ for each $k<k_0$. We analyse $RR(k)$ assuming that  $\|{\bf w}\|_2\leq \epsilon_{omp}$.   Applying  the triangle inequality $\|{\bf a}+{\bf b}\|_2\leq \|{\bf a}\|_2+\|{\bf b}\|_2$, the reverse triangle inequality $\|{\bf a}+{\bf b}\|_2\geq \|{\bf a}\|_2-\|{\bf b}\|_2$  and the bound $\|({\bf I}_n-{\bf P}_{k}){\bf w}\|_2\leq \|{\bf w}\|_2 $ to  $\|{\bf r}^{k}\|_2=\|({\bf I}_n-{\bf P}_{k}){\bf X}\boldsymbol{\beta}+({\bf I}_n-{\bf P}_{k}){\bf w}\|_2$  gives 
\begin{equation}\label{aaa1}
\|({\bf I}_n-{\bf P}_{k}){\bf X}\boldsymbol{\beta}\|_2-\|{\bf w}\|_2\leq \|{\bf r}^{k}\|_2\leq \|({\bf I}_n-{\bf P}_{k}){\bf X}\boldsymbol{\beta}\|_2+\|{\bf w}\|_2.
\end{equation}
Let ${u^k}=\mathcal{S}/\mathcal{S}_k$  denotes the indices in $\mathcal{S}$ that are not selected after the $k^{th}$ iteration. 
Note that ${\bf X}\boldsymbol{\beta}={\bf X}_{\mathcal{S}}\boldsymbol{\beta}_{\mathcal{S}}={\bf X}_{\mathcal{S}_k}\boldsymbol{\beta}_{\mathcal{S}_k}+{\bf X}_{u^k}\boldsymbol{\beta}_{u^k}$. Since, ${\bf I}_n-{\bf P}_k$ projects orthogonal to the column space $span({\bf X}_{\mathcal{S}_k})$, $({\bf I}_n-{\bf P}_{k}){\bf X}_{\mathcal{S}_k}\boldsymbol{\beta}_{\mathcal{S}_k}={\bf 0}_n$.
Hence, $({\bf I}_n-{\bf P}_{k}){\bf X}\boldsymbol{\beta}=({\bf I}_n-{\bf P}_{k}){\bf X}_{u^k}\boldsymbol{\beta}_{u^k}$. Further, $\mathcal{S}_k\subset \mathcal{S}$ implies that  $card(\mathcal{S}_k)+card(u^k)= k_0$ and $\mathcal{S}_k\cap u^k=\phi$. Hence, by  Lemma 2 of \cite{latest_omp},
\begin{equation}\label{aaa2}
\sqrt{1-\delta_{k_0}}\|\boldsymbol{\beta}_{u^k}\|_2\leq \|({\bf I}_n-{\bf P}_{k}){\bf X}_{u^k}\boldsymbol{\beta}_{u^k}\|_2\leq \sqrt{1+\delta_{k_0}}\|\boldsymbol{\beta}_{u^k}\|_2.
\end{equation} 
Substituting (\ref{aaa2}) in (\ref{aaa1}) gives
\begin{equation}\label{Caibound}
\sqrt{1-\delta_{k_0}}\|\boldsymbol{\beta}_{u^k}\|_2-\|{\bf w}\|_2\leq\|{\bf r}^{k}\|_2\leq \sqrt{1+\delta_{k_0}}\|\boldsymbol{\beta}_{u^k}\|_2+\|{\bf w}\|_2.
\end{equation}
Note that $\boldsymbol{\beta}_{u^{k-1}}=\boldsymbol{\beta}_{u^{k}}+\boldsymbol{\beta}_{u^{k-1}/u^{k}}$
after appending enough zeros in appropriate locations. $\boldsymbol{\beta}_{u^{k-1}/u^{k}}$  has only one non zero entry.  Hence, $\|\boldsymbol{\beta}_{u^{k-1}/u^{k}}\|_2\leq \boldsymbol{\beta}_{max}$. Applying triangle inequality to $\boldsymbol{\beta}_{u^{k-1}}=\boldsymbol{\beta}_{u^{k}}+\boldsymbol{\beta}_{u^{k-1}/u^{k}}$ gives
the bound 
\begin{equation}\label{temp_bound}
\|\boldsymbol{\beta}_{u^{k-1}}\|_2\leq \|\boldsymbol{\beta}_{u^{k}}\|_2+\|\boldsymbol{\beta}_{u^{k-1}/u^{k}}\|_2 \leq  \|\boldsymbol{\beta}_{u^{k}}\|_2+  \boldsymbol{\beta}_{max}
\end{equation}
Applying (\ref{temp_bound}) and (\ref{Caibound}) in $RR(k)$ for $k<k_0$ gives
\begin{equation}\label{A1bound}
\begin{array}{ll}
RR(k)=\dfrac{\|{\bf r}^{k}\|_2}{\|{\bf r}^{k-1}\|_2} &\geq \dfrac{\sqrt{1-\delta_{k_0}}\|\boldsymbol{\beta}_{u^k}\|_2-\|{\bf w}\|_2}{\sqrt{1+\delta_{k_0}}\|\boldsymbol{\beta}_{u^{k-1}}\|_2+\|{\bf w}\|_2}\\
&\geq \dfrac{\sqrt{1-\delta_{k_0}}\|\boldsymbol{\beta}_{u^k}\|_2-\|{\bf w}\|_2}{\sqrt{1+\delta_{k_0}}\left[\|\boldsymbol{\beta}_{u^{k}}\|_2+\boldsymbol{\beta}_{max}\right]+\|{\bf w}\|_2}\\
\end{array}
\end{equation}
whenever $\|{\bf w}\|_2\leq \epsilon_{omp}$. The R.H.S of (\ref{A1bound}) can be rewritten as
\begin{equation}\label{A1bound2}
\begin{array}{ll}
\dfrac{\sqrt{1-\delta_{k_0}}\|\boldsymbol{\beta}_{u^k}\|_2-\|{\bf w}\|_2}{\sqrt{1+\delta_{k_0}}
\left[\|\boldsymbol{\beta}_{u^{k}}\|_2+\boldsymbol{\beta}_{max}\right]+\|{\bf w}\|_2}=\dfrac{\sqrt{1-\delta_{k_0}}}{\sqrt{1+\delta_{k_0}}}  \\  \ \ \ \ \ \ \ \ \ \ \ 
 \left(1-\dfrac{\dfrac{\|{\bf w}\|_2}{\sqrt{1-\delta_{k_0}}}+\dfrac{\|{\bf w}\|_2}{\sqrt{1+\delta_{k_0}}}+\boldsymbol{\beta}_{max}}{\|\boldsymbol{\beta}_{u^{k}}\|_2+\boldsymbol{\beta}_{max}+\dfrac{\|{\bf w}\|_2}{\sqrt{1+\delta_{k_0}}}}\right)
\end{array}
\end{equation}
From (\ref{A1bound2}) it is clear that the R.H.S of (\ref{A1bound}) decreases with decreasing $\|\boldsymbol{\beta}_{u^k}\|_2$.  Note  that the minimum value of $\|\boldsymbol{\beta}_{u^{k}}\|_2$ is $\boldsymbol{\beta}_{min}$ itself. Hence, substituting $\|\boldsymbol{\beta}_{u^{k}}\|_2\geq \boldsymbol{\beta}_{min}$ in (\ref{A1bound})   gives 
\begin{equation}\label{lb_on_klessk0_temp}
RR(k)\geq \dfrac{\sqrt{1-\delta_{k_0}}\boldsymbol{\beta}_{min}-\|{\bf w}\|_2}{\sqrt{1+\delta_{k_0}}(\boldsymbol{\beta}_{max}+\boldsymbol{\beta}_{min})+\|{\bf w}\|_2},
\end{equation}
$\forall k<k_0$ whenever  $\|{\bf w}\|_2<\epsilon_{omp}$.  
This along with the fact $RR(k)\geq 0$ implies that 
\begin{equation}
RR(k)\geq \dfrac{\sqrt{1-\delta_{k_0}}\boldsymbol{\beta}_{min}-\|{\bf w}\|_2}{\sqrt{1+\delta_{k_0}}(\boldsymbol{\beta}_{max}+\boldsymbol{\beta}_{min})+\|{\bf w}\|_2} \mathcal{I}_{\{\|{\bf w}\|_2\leq \epsilon_{omp}\}},
\end{equation}
where $\mathcal{I}_{\{\mathcal{E}\}}$ is the indicator function returning one when the event $\mathcal{E}$ occurs and zero otherwise.  
Note that $\|{\bf w}\|_2\overset{P}{\rightarrow }0$ as $\sigma^2 \rightarrow 0$. This implies that $\dfrac{\sqrt{1-\delta_{k_0}}\boldsymbol{\beta}_{min}-\|{\bf w}\|_2}{\sqrt{1+\delta_{k_0}}(\boldsymbol{\beta}_{max}+\boldsymbol{\beta}_{min})+\|{\bf w}\|_2} \overset{P}{\rightarrow } \dfrac{\sqrt{1-\delta_{k_0}}\boldsymbol{\beta}_{min}}{\sqrt{1+\delta_{k_0}}(\boldsymbol{\beta}_{max}+\boldsymbol{\beta}_{min})}$ and $\mathcal{I}_{\{\|{\bf w}\|_2\leq \epsilon_{omp}\}} \overset{P}{\rightarrow}1$. Substituting these bounds in (\ref{lb_on_klessk0_temp}) one can obtain  $\underset{\sigma^2\rightarrow 0}{\lim}\mathbb{P}\left(RR(k)>\dfrac{\sqrt{1-\delta_{k_0}}\boldsymbol{\beta}_{min}}{\sqrt{1+\delta_{k_0}}(\boldsymbol{\beta}_{max}+\boldsymbol{\beta}_{min})}\right)=1$. 
\end{proof}
\section*{Appendix C: Proof of Theorem \ref{thm:rrm_finite} }
 \begin{proof}
 For RRM support estimate $\mathcal{S}_{RRM}=\mathcal{S}_{k_{RRM}}$ where $k_{RRM}=\underset{k}{\arg\min}RR(k)$ to satisfy $\mathcal{S}_{RRM}=\mathcal{S}$, it is sufficient that the  following four events $\mathcal{A}_1$, $\mathcal{A}_2$, $\mathcal{A}_3$  and $\mathcal{A}_4$ occur simultaneously. \\
 $\mathcal{A}_1=\{k_{min}=k_{0}\}$. \\
 $\mathcal{A}_2=\{RR(k)>RR(k_0),\forall k<k_{0}\}$. \\
 $\mathcal{A}_3=\{RR(k_{0})<\underset{j=1,\dotsc,k_{max}}{\min}\Gamma_{RRT}^{\alpha}(j)\}$. \\
 $\mathcal{A}_4=\{RR(k)>\underset{j=1,\dotsc,k_{max}}{\min}\Gamma_{RRT}^{\alpha}(j),\forall k>k_{min} \}$.
 
This is explained as follows.  Event $\mathcal{A}_1$ true implies that $\mathcal{S}_{k_{0}}=\mathcal{S}$ and $k_{min}=k_{0}$.  
$\mathcal{A}_1 \cap \mathcal{A}_2$ is true implies that $k_{RRM}\geq k_{min}$, i.e.,  $k_{RRM}$ will not underestimate $k_{min}$. Event $\mathcal{A}_3\cap  \mathcal{A}_4$ implies that $RR(k_{min})>RR(k)$ for all $k>k_{min}$  which ensures that  $k_{RRM}\leq k_{min}$, i.e., $k_{RRM}$ will not overestimate $k_{min}$. Hence, $\mathcal{A}_2\cap \mathcal{A}_3\cap \mathcal{A}_4$ implies that $k_{RRM}=k_{min}$. This together with $\mathcal{A}_1$  implies that $k_{RRM}=k_0$ and $\mathcal{S}_{RRM}=\mathcal{S}$.  Hence, $\mathbb{P}(\mathcal{S}_{RRM}=\mathcal{S})\geq \P\left(\mathcal{A}_1\cap \mathcal{A}_2\cap \mathcal{A}_3\cap \mathcal{A}_4\right)$.  

By Lemma \ref{lemma:latest_omp}, it is true that $\mathcal{A}_1=\{k_{min}=k_{0}\}$ is true  once $\|{\bf w}\|_2\leq \epsilon_{omp}$. { Using the bound $\sqrt{1-\delta_{k_0}}\|\boldsymbol{\beta}_{u^k}\|_2-\|{\bf w}\|_2\leq\|{\bf r}^{k}\|_2\leq \sqrt{1+\delta_{k_0}}\|\boldsymbol{\beta}_{u^k}\|_2+\|{\bf w}\|_2$ from (\ref{Caibound}) in the proof of Lemma \ref{lemma:kless} and the fact that $u_{k_0}=\emptyset$, we have $\|{\bf r}^{k_0}\|_2\leq \|{\bf w}\|_2$  and  $ \|{\bf r}^{k_0-1}\|_2 \geq \sqrt{1-\delta_{k_0}}\|\boldsymbol{\beta}_{u^{k_0-1}}\|_2-\|{\bf w}\|_2  \geq  \sqrt{1-\delta_{k_0}}\boldsymbol{\beta}_{min}-\|{\bf w}\|_2 $. } Substituting these bounds in $RR(k_0)=\frac{\|{\bf r}^{k_0}\|_2}{\|{\bf r}^{k_0-1}\|_2}$  gives
\begin{equation}
RR(k_{0})=RR(k_{min})\leq \dfrac{\|{\bf w}\|_2}{\sqrt{1-\delta_{k_{0}}}\boldsymbol{\beta}_{min}-\|{\bf w}\|_2},
\end{equation}
once $\|{\bf w}\|_2 \leq \epsilon_{omp}$. Hence the event $\mathcal{A}_3$, i.e., $RR(k_{0})<\underset{j=1,\dotsc,k_{max}}{\min}\Gamma_{RRT}^{\alpha}(j)$ is true once  
\begin{equation}
\dfrac{\|{\bf w}\|_2}{\sqrt{1-\delta_{k_{0}}}\boldsymbol{\beta}_{min}-\|{\bf w}\|_2}<\underset{j=1,\dotsc,k_{max}}{\min}\Gamma_{RRT}^{\alpha}(j)
\end{equation}
which in turn is true once $\|{\bf w}\|_2\leq \min(\epsilon_{omp},\tilde{\epsilon}_{rrt})$. 
 
Next we consider the event $\mathcal{A}_2$. From (\ref{lb_on_klessk0_temp}), we have $RR(k)\geq \dfrac{\sqrt{1-\delta_{k_0}}\boldsymbol{\beta}_{min}-\|{\bf w}\|_2}{\sqrt{1+\delta_{k_0}}(\boldsymbol{\beta}_{max}+\boldsymbol{\beta}_{min})+\|{\bf w}\|_2}, \ \forall k<k_0$ whenever  $\|{\bf w}\|_2<\epsilon_{omp}$. Hence, $RR(k_{0})<\underset{k<k_{0}}{\min}RR(k)$ is true once the lower bound $\dfrac{\sqrt{1-\delta_{k_0}}\boldsymbol{\beta}_{min}-\|{\bf w}\|_2}{\sqrt{1+\delta_{k_0}}(\boldsymbol{\beta}_{max}+\boldsymbol{\beta}_{min})+\|{\bf w}\|_2}$ on $RR(k)$ for $k<k_0$ is higher than the upper bound $\dfrac{\|{\bf w}\|_2}{\sqrt{1-\delta_{k_{0}}}\boldsymbol{\beta}_{min}-\|{\bf w}\|_2}$ on $RR(k_0)$. This is true  once $\|{\bf w}\|_2\leq \min(\epsilon_{omp},\epsilon_{rrm})$. Consequently, events $\mathcal{A}_1\cap\mathcal{A}_2\cap\mathcal{A}_3$ occur simultaneously once $\|{\bf w}\|_2\leq \min(\epsilon_{omp},\tilde{\epsilon_{rrt}},\epsilon_{rrm})$. Since $\epsilon_{\sigma}=\sigma\sqrt{n+2\sqrt{n\log(n)}}$ satisfies $\P(\|{\bf w}\|_2\leq \epsilon_{\sigma})\geq 1-1/n$, it is true that $\P(\mathcal{A}_1\cap\mathcal{A}_2\cap\mathcal{A}_3)\geq 1-1/n$ once 
$\epsilon_{\sigma}\leq \min(\epsilon_{omp},\tilde{\epsilon_{rrt}},\epsilon_{rrm})$.  

Next we consider the event $\mathcal{A}_4$.  Following Lemma \ref{lemma:RR_properties}, it is true that 
\begin{equation}
\begin{array}{ll}
\mathbb{P}(RR(k)>\underset{j=1,\dotsc,k_{max}}{\min}\Gamma_{RRT}^{\alpha}(j),\forall k>k_{min})\\ 
\ \ \ \ \ \ \ \ \ \ \ \geq \P(RR(k)>\Gamma_{RRT}^{\alpha}(k),\forall k>k_{min})\geq 1-\alpha, 
\end{array}
\end{equation}
for all $\sigma^2>0$.
 Hence,  the event $\mathcal{A}_4$ occurs with probability atleast $1-\alpha$, $\forall \sigma^2>0$. 

 Combining all these results give $\mathbb{P}(\mathcal{S}_{RRM}=\mathcal{S})\geq 1-1/n-\alpha$ whenever $\epsilon_{\sigma}\leq \min(\tilde{\epsilon_{rrt}},\epsilon_{rrm},\epsilon_{omp})$.

 \end{proof}

 \section*{Appendix D: Proof of Theorem \ref{thm:rmm_hsc} }
 \begin{proof}
 To prove that $\underset{\sigma^2\rightarrow 0}{\lim}\mathbb{P}(\mathcal{S}_{RRM}=\mathcal{S})=1$, it is sufficient to show that for every fixed $\delta>0$, there exists a $\sigma^2_{\delta}>0$ such that $\mathbb{P}(\mathcal{S}_{RRM}=\mathcal{S})\geq 1-\delta$ for all $\sigma^2<\sigma^2_{\delta}$.  Consider the events $\{\mathcal{A}_j\}_{j=1}^4$ with the same definition as in Appendix C. Then $\mathbb{P}(\mathcal{S}_{RRM}=\mathcal{S})\geq \mathbb{P}(\mathcal{A}_1\cap \mathcal{A}_2\cap \mathcal{A}_3\cap \mathcal{A}_4)$. Let $\delta>0$ be any given number.  Fix the alpha parameter $\alpha=\frac{\delta}{2}$. Applying  Lemma \ref{lemma:RR_properties} with $\alpha=\frac{\delta}{2}$ gives the bound 
 \begin{equation}\label{eq:rrm_hsc1}
 \mathbb{P}(\mathcal{A}_4)=\mathbb{P}(RR(k)>\underset{j=1,\dotsc,k_{max}}{\min}\Gamma_{RRT}^{\frac{\delta}{2}}(j), \forall k>k_{min})\geq 1-\frac{\delta}{2},
 \end{equation}
  for all $\sigma^2>0$.  Following the proof of Theorem \ref{thm:rrm_finite}, we have 
\begin{equation}
\mathbb{P}(\mathcal{A}_1\cap \mathcal{A}_2\cap \mathcal{A}_3)\geq \mathbb{P}\left(\|{\bf w}\|_2\leq \min(\epsilon_{rrm},\tilde{\epsilon}_{rrt},\epsilon_{omp})\right).
\end{equation} 
{ Note that both $\epsilon_{rrm}>0$ and $\epsilon_{omp}>0$ are both independent of $\alpha$  and hence $\delta$. At the same time,  $\tilde{\epsilon}_{rrt}=\dfrac{\underset{1<k\leq k_{max}}{\min}\Gamma_{RRT}^{\alpha}(k)\sqrt{1-\delta_{{k_{0}}}}\boldsymbol{\beta}_{min}}
{1+\underset{1\leq k\leq k_{max}}{\min}\Gamma_{RRT}^{\alpha}(k)}$ is dependent on $\alpha$ and hence $\delta$.   Since $\mathbb{B}(a,b)$ is a continuous random variable with support in $(0,1)$, for every $z>0$, $a>0$ and $b>0$ , $F^{-1}_{a,b}(z)>0$. Hence, for each $\delta>0$,  $\Gamma_{RRT}^{\frac{\delta}{2}}(k)>0$ which implies that $\tilde{\epsilon}_{rrt}>0$. This inturn implies that   $\min(\epsilon_{rrm},\tilde{\epsilon}_{rrt},\epsilon_{omp})>0$.} Note that $\|{\bf w}\|_2\overset{P}{\rightarrow }0$ as $\sigma^2\rightarrow 0$. This implies that for every fixed $\delta>0$, $\exists \sigma^2(\delta)>0$ such that 
\begin{equation}\label{eq:rrm_hsc2}
\begin{array}{ll}
\mathbb{P}(\mathcal{A}_1\cap \mathcal{A}_2\cap \mathcal{A}_3)&\geq \mathbb{P}\left(\|{\bf w}\|_2\leq \min(\epsilon_{rrm},\tilde{\epsilon}_{rrt},\epsilon_{omp})\right)\\ 
&\geq 1-\frac{\delta}{2}
\end{array}
\end{equation}
 for all $\sigma^2<\sigma^2(\delta)$. Combining (\ref{eq:rrm_hsc1}) and (\ref{eq:rrm_hsc2}), one can obtain $\mathbb{P}(\mathcal{A}_1\cap \mathcal{A}_2\cap \mathcal{A}_3\cap \mathcal{A}_4)\geq 1-\delta$ for all $\sigma^2<\sigma^2(\delta)$. Since this is true for all $\delta>0$, we have $\underset{\sigma^2\rightarrow 0}{\lim}\mathbb{P}(S_{RRM}=\mathcal{S})=1$. 
 
 \end{proof}
 \section*{Appendix E: Proof of Theorem \ref{thm:RRTA} }
 \begin{proof}
{Define the events $\mathcal{E}_1=\{\mathcal{S}_{k_0}=\mathcal{S}\}=\{k_{min}=k_0\}$, $\mathcal{E}_2=\{RR(k_0)<\Gamma_{RRT}^{\alpha^*}(k_0)\}$ and $\mathcal{E}_3=\{RR(k)>\Gamma_{RRT}^{\alpha^*}(k),\forall k>k_{min} \}$.   Event $\mathcal{E}_1 \cap \mathcal{E}_2$ implies that the RRTA estimate $k_{RRTA}=\max\{k:RR(k)<\Gamma_{RRT}^{\alpha^*}(k)\}$ satisfies $k_{RRTA}\geq k_{min}$, whereas,   the event $\mathcal{E}_2 \cap \mathcal{E}_3$ implies that the RRTA estimate $k_{RRTA}\leq k_{min}$. Hence, Event $\mathcal{E}_1 \cap \mathcal{E}_2\cap \mathcal{E}_3$ implies that $k_{RRTA}=k_{min}=k_0$ and $\mathcal{S}_{RRTA}=\mathcal{S}$. Hence $\mathbb{P}(\mathcal{S}_{RRTA}=\mathcal{S})\geq \mathbb{P}(\mathcal{E}_1 \cap \mathcal{E}_2\cap \mathcal{E}_3)$.  }

By Lemma \ref{lemma:latest_omp}, $\mathcal{E}_1$ is true once $\|{\bf w}\|_2\leq \epsilon_{omp}$. This along with $\|{\bf w}\|_2\overset{P}{\rightarrow} 0$ as $\sigma^2\rightarrow 0$ implies that $\underset{\sigma^2\rightarrow 0}{\lim}\mathbb{P}(\mathcal{E}_1)=1$. Next we consider $\mathcal{E}_2$. By the definition of $\mathcal{E}_2$
\begin{equation}
\mathbb{P}(\mathcal{E}_2)=\mathbb{P}\left(\dfrac{\Gamma_{RRT}^{\alpha^*}(k)}{\underset{k}{\min}RR(k)}\dfrac{\underset{k}{\min}RR(k)}{RR(k_0)}>1\right)
\end{equation}
Following Theorem \ref{thm:rmm_hsc} and it's proof, we know that $\underset{k}{\min}RR(k)\overset{P}{\rightarrow} RR(k_0)$ as $\sigma^2\rightarrow 0$.  Hence, $\dfrac{\underset{k}{\min}RR(k)}{RR(k_0)}\overset{P}{\rightarrow }1  $ as $\sigma^2\rightarrow 0$. 
From Theorem \ref{thm:rmm_hsc}, we also know that $\underset{k}{\min}RR(k)\overset{P}{\rightarrow} 0$ as $\sigma^2\rightarrow 0$. Since the function $\alpha^*(x)=\min(PFD_{finite},x^q)$ is continuous around $x=0$ for every $q>0$ and $PFD_{finite}>0$, this implies\footnote{Suppose that a R.V $Z \overset{P}{\rightarrow} c$ and $g(x)$ is a function continuous at $x=c$. Then $g(Z)\overset{P}{\rightarrow}g(c)$\cite{wasserman2013all}. } that $\alpha^*=\min(PFD_{finite},\underset{k}{\min}RR(k)^q)\overset{P}{\rightarrow }0$ as $\sigma^2\rightarrow 0$.  
\begin{lemma}\label{lemma:Gamma} For any function $f(x)\rightarrow 0$ as $x\rightarrow 0$, $\Gamma_{RRT}^{f(x)}(k_0)/x \rightarrow \infty$ as $x\rightarrow 0$ once $f(x)^{\frac{2}{n-k_0}}/x^2 \rightarrow \infty$ as $x\rightarrow 0$.
\end{lemma}      
\begin{proof}
Please see Appendix F.
\end{proof}
Please note that the function $f(x)=\min(PFD_{finite},x^q)$ satisfies $f(x)^{\frac{2}{n-k_0}}/x^2 \rightarrow \infty$ as $x \rightarrow 0$ once $2q/(n-k_0)<2$ which is  true once $n>k_0+q$.  Since the function  $\alpha^*=f(RR(k))=\min(PFD_{finite},RR(k)^q)$ is continuous around zero and  $\underset{k}{\min}RR(k)\overset{P}{\rightarrow }0$  as $\sigma^2 \rightarrow 0$, $\dfrac{\Gamma_{RRT}^{\alpha^*}(k)}{\underset{k}{\min}RR(k)}\overset{P}{\rightarrow} \infty$  as $\sigma^2\rightarrow 0$ once $n>k_0+q$. Since   $\dfrac{\Gamma_{RRT}^{\alpha^*}(k)}{\underset{k}{\min}RR(k)}\overset{P}{\rightarrow }\infty$ and   $\dfrac{\underset{k}{\min}RR(k)}{RR(k_0)} \overset{P}{\rightarrow} 1$, we have  $\underset{\sigma^2\rightarrow 0}{\lim}\mathbb{P}\left(\dfrac{\Gamma_{RRT}^{\alpha^*}(k)}{\underset{k}{\min}RR(k)}\dfrac{\underset{k}{\min}RR(k)}{RR(k_0)}>1\right)=1$ and $\underset{\sigma^2\rightarrow 0}{\lim}\mathbb{P}(\mathcal{E}_2)=1$.  

Next we consider the event $\mathcal{E}_3=\{RR(k)>\Gamma_{RRT}^{\alpha^*}(k),\forall k>k_{min} \}$. Please note that the bound  $\mathbb{P}(RR(k)>\Gamma_{RRT}^{\alpha}(k), \forall k>k_{min})\geq 1-\alpha$ for all $\sigma^2>0$ in Lemma \ref{lemma:RR_properties} is derived assuming that $\alpha$ is a deterministic quantity. However, $\alpha^*=\min(PFD_{finite},\underset{k}{\min} RR(k)^q)$ in RRTA  is a stochastic quantity and hence Lemma \ref{lemma:RR_properties} is not directly applicable. Note that  for any $\delta>0$, we have
\begin{equation}
\begin{array}{ll}
\mathbb{P}(RR(k)>\Gamma_{RRT}^{\alpha^*}(k),\forall k>k_{min} )\\  
\ \ \ \ \ \ \ \ \ \ \ \  \overset{(a)}{\geq} \mathbb{P}(\{RR(k)>\Gamma_{RRT}^{\alpha^*}(k),\forall k>k_{min}\} \cap \{\alpha^*\leq \delta\} )\\
\ \ \ \ \ \ \ \ \ \ \ \ \   \overset{(b)}{\geq} \mathbb{P}(\{RR(k)>\Gamma_{RRT}^{\delta}(k),\forall k>k_{min}\} \cap \{\alpha^*\leq \delta\} )
\end{array}  
\end{equation}
(a) follows from the intersection bound $\mathbb{P}(A\cap B)\geq \mathbb{P}(A)$.  Note that $F^{-1}_{a,b}(z)$ is a monotonically increasing function of $z$. This implies $\Gamma_{RRT}^{\alpha}(k)<\Gamma_{RRT}^{\delta}(k)$ when $\alpha<\delta$. (b) follows from this. 

 Note that by Lemma \ref{lemma:RR_properties}, we have $\mathbb{P}(\{RR(k)>\Gamma_{RRT}^{\delta}(k),\forall k>k_{min}\})\geq 1-\delta$ for all $\sigma^2>0$.  Further, $\alpha^*\overset{P}{\rightarrow} 0$ as $\sigma^2\rightarrow $  implies that $\underset{\sigma^2\rightarrow 0}{\lim}\mathbb{P}({\alpha^*\leq \delta})=1$. 
These two results together imply  $\underset{\sigma^2\rightarrow 0}{\lim}\mathbb{P}(\{RR(k)>\Gamma_{RRT}^{\alpha^*}(k),\forall k>k_{min}\} \cap \{\alpha^*\leq \delta\} )\geq 1-\delta$. 
Since this is true for all $\delta>0$, we have $\underset{\sigma^2\rightarrow 0}{\lim}\mathbb{P}(\{RR(k)>\Gamma_{RRT}^{\alpha^*}(k),\forall k>k_{min}\} \cap \{\alpha^*\leq \delta\} )=1$ which in turn imply $\underset{\sigma^2\rightarrow 0}{\lim}\mathbb{P}(\mathcal{E}_3)=\underset{\sigma^2\rightarrow 0}{\lim}\mathbb{P}(RR(k)>\Gamma_{RRT}^{\alpha^*}(k),\forall k>k_{min} )=1$.

Since $\underset{\sigma^2\rightarrow 0}{\lim}\mathbb{P}(\mathcal{E}_j)=1$ for $j=1,2$ and $3$, it follows that $\underset{\sigma^2\rightarrow 0}{\lim}\mathbb{P}(\mathcal{S}_{RRTA}=\mathcal{S})\geq \underset{\sigma^2\rightarrow 0}{\lim}\mathbb{P}(\mathcal{E}_1 \cap \mathcal{E}_2\cap \mathcal{E}_3)=1$. 

 \end{proof}
 \section*{Appendix F: Proof of Lemma \ref{lemma:Gamma} }
 \begin{proof}
 Expanding  $F^{-1}_{a,b}(z)$  at $z=0$ using the expansion given in [http://functions.wolfram.com/06.23.06.0001.01] gives
\begin{equation}\label{beta_exp}
\begin{array}{ll}
F^{-1}_{a,b}(z)=\rho(n,1)+\dfrac{b-1}{a+1}\rho(n,2) \\
+\dfrac{(b-1)(a^2+3ab-a+5b-4)}{2(a+1)^2(a+2)}\rho(n,3)
+O(z^{(4/a)})
\end{array}
\end{equation}
for all $a>0$. Here $\rho(n,l)=(az{B}(a,b))^{(l/a)}$.  Note that $\Gamma_{RRT}^{f(x)}(k_0)=\sqrt{F^{-1}_{\frac{n-k_0}{2},{\frac{1}{2}}}\left(\dfrac{f(x)}{k_{max}(p-k_0+1)}\right)}$.  We associate $a=\frac{n-k_0}{2}$, $b=1/2$ , $z=\dfrac{f(x)}{k_{max}(p-k_0+1)}$ and $\rho(n,l)=(az{B}(a,b))^{(l/a)}=\left(\frac{\left(\frac{n-k_0}{2}\right)f(x){B}\left(\frac{n-k_0}{2},0.5\right)}{k_{max}(p-k_0+1)}\right)^{\frac{2l}{n-k_0}}$ for $l\geq 1$.    $\dfrac{\Gamma_{RRT}^{f(x)}(k_0)}{x}=\sqrt{\dfrac{F^{-1}_{\frac{n-k_0}{2},\frac{1}{2}}\left(\dfrac{f(x)}{k_{max}(p-k_0+1)}\right)}{x^2}}$. {  Note that the term $f(x)^{\frac{2l}{n-k_0}}$ is the only term in $\rho(n.l)$ that depends on $x$.  Now from the expansion and the fact that $\underset{x \rightarrow 0}{\lim}f(x)^{\frac{2l}{n-k_0}}/f(x)^{\frac{2}{n-k_0}}=0$ for each $l>1$, it is clear that $\sqrt{\dfrac{F^{-1}_{\frac{n-k_0}{2},\frac{1}{2}}\left(\dfrac{f(x)}{k_{max}(p-k_0+1)}\right)}{x^2}} \rightarrow \infty$ as $x \rightarrow 0$ once  $f(x)^{\frac{2}{n-k_0}}/x^2\rightarrow \infty$. }
 \end{proof}
\bibliographystyle{IEEEtran}
\bibliography{compressive}

%
%
%
\end{document}

%% file: notation.tex
\def\P{\mathbb{P}}

\newcommand{\upperRomannumeral}[1]{\uppercase\expandafter{\romannumeral#1}}

\newtheorem{lemma}{Lemma}{}
  \newtheorem{thm}{Theorem}